\documentclass[11pt]{article}
\usepackage[margin=0.8in]{geometry}

\usepackage{varioref}
\usepackage[usenames,svgnames,xcdraw,table]{xcolor}
\definecolor{DarkBlue}{rgb}{0.1,0.1,0.5}
\definecolor{DarkGreen}{rgb}{0.1,0.5,0.1}

\usepackage{hyperref}
\hypersetup{
   colorlinks   = true,
 linkcolor    = DarkBlue, 
 urlcolor     = DarkBlue, 
	 citecolor    = DarkGreen 
}

\usepackage{appendix}

\usepackage{mathtools}

\usepackage{algorithmic}

\usepackage{algorithm}
\usepackage{array}

\usepackage[font={small,it}]{caption}

\usepackage{graphicx,epsfig,amsmath,latexsym,amssymb,verbatim}

\usepackage[square,sort,semicolon,numbers]{natbib}

\newcommand{\extra}[1]{}

\usepackage{amsmath,amssymb,amsfonts}
\usepackage{amsthm}
\usepackage{thmtools,thm-restate}
\usepackage{enumitem}

\newtheorem{theorem}{Theorem}

\newtheorem{lemma}[theorem]{Lemma}

\newtheorem{claim}{Claim}
\newtheorem*{remark2}{Remark}

\def\squareforqed{\hbox{\rlap{$\sqcap$}$\sqcup$}}
\def\qed{\ifmmode\squareforqed\else{\unskip\nobreak\hfil
\penalty50\hskip1em\null\nobreak\hfil\squareforqed
\parfillskip=0pt\finalhyphendemerits=0\endgraf}\fi}
\def\endenv{\ifmmode\;\else{\unskip\nobreak\hfil
\penalty50\hskip1em\null\nobreak\hfil\;
\parfillskip=0pt\finalhyphendemerits=0\endgraf}\fi}

\renewenvironment{proof}{\noindent \textbf{{Proof~} }}{\qed\medskip}
\newenvironment{proof+}[1]{\noindent \textbf{{Proof #1~} }}{\qed\medskip}

\mathchardef\ordinarycolon\mathcode`\:
\mathcode`\:=\string"8000
\def\vcentcolon{\mathrel{\mathop\ordinarycolon}}
\begingroup \catcode`\:=\active
  \lowercase{\endgroup
  \let :\vcentcolon
  }

\newcommand{\nc}{\newcommand}

\DeclareMathOperator*{\argmin}{arg\,min}
\DeclareMathOperator*{\argmax}{arg\,max}

\nc{\barA}{\overline{A}}
\nc{\barB}{\overline{B}}
\nc{\barC}{\overline{C}}
\nc{\barD}{\overline{D}}
\nc{\barR}{\overline{R}}
\nc{\barX}{\overline{X}}
\nc{\barY}{\overline{Y}}
\nc{\barU}{\overline{U}}

\newcommand{\EF}{\mathrm{EF}}
\newcommand{\NSW}{\mathrm{NSW}}
\newcommand{\SW}{\mathrm{SW}}

\newcommand{\NP}{\rm{NP}}
\newcommand{\APX}{\rm{APX}}

\newcommand{\ALG}{\textsc{Alg}}
\newcommand{\JISP}{\textsc{JISP}}

\newcommand{\I}{\mathcal{I}}

\newcommand{\C}{\mathcal{C}}

\begin{document}

\title{{\bfseries Fair and Efficient Cake Division with Connected Pieces}}
\author{Eshwar Ram Arunachaleswaran\thanks{Indian Institute of Science. {\tt eshwarram.arunachaleswaran@gmail.com}} \quad Siddharth Barman\thanks{Indian Institute of Science. {\tt barman@iisc.ac.in}} \quad Rachitesh Kumar\thanks{Indian Institute of Science. {\tt rachiteshkumar@gmail.com}} \quad Nidhi Rathi\thanks{Indian Institute of Science. {\tt nidhirathi@iisc.ac.in}}}

\date{}
\maketitle
\begin{abstract}
The classic cake-cutting problem provides a model for addressing fair and efficient allocation of a divisible, heterogeneous resource (metaphorically, the cake) among agents with distinct preferences. Focusing on a standard formulation of cake cutting, in which each agent must receive a contiguous piece of the cake, this work establishes algorithmic and hardness results for multiple fairness/efficiency measures. 

First, we consider the well-studied notion of envy-freeness and develop an efficient algorithm that finds a cake division (with connected pieces) wherein the envy is multiplicatively within a factor of $2 + o(1)$. The same algorithm in fact achieves an approximation ratio of $3 + o(1)$ for the problem of finding cake divisions with as large a Nash social welfare ($\NSW$) as possible. $\NSW$ is another standard measure of fairness and this work also establishes a connection between envy-freeness and $\NSW$: approximately envy-free cake divisions (with connected pieces) always have near-optimal Nash social welfare. Furthermore, we  develop an approximation algorithm for maximizing the $\rho$-mean welfare--this unifying objective, with different values of $\rho$, interpolates between notions of fairness ($\NSW$) and efficiency (average social welfare). Finally, we complement these algorithmic results by proving that maximizing $\NSW$ (and, in general, the $\rho$-mean welfare) is $\APX$-hard in the cake-division context. 
\end{abstract}

\section{Introduction}
Cake cutting is a fundamental problem in the fair-division literature. It models the task of allocating a divisible, heterogeneous resource among agents with distinct preferences, but equal entitlements. Indeed, the classic work of Steinhaus, Banach, and Knaster \cite{S48problem}---which lays the mathematical foundations of fair division---addresses cake cutting. Over the years, this problem has not only inspired the development of many interesting mathematical connections and algorithms (see, e.g.,~\cite{robertson1998cake}), but has also been found relevant in real-world settings, such as border negotiations and divorce settlements \cite{brams1996fair}. 
Implementations of cake-division methods on platforms (such as Adjusted Winner~\cite{AdjustedWinner}) further substantiate the practical relevance of this framework.  

Here, the cake is represented by the segment $[0,1]$ and the cardinal preferences of the agents are specified via valuation functions over (the intervals of) the cake. We will throughout focus on the setting wherein the cake needs to be partitioned into exactly $n$ connected pieces (intervals) and each of the $n$ agents receives one of these intervals. This is a well-studied formulation of cake cutting and is motivated by applications wherein connectivity (across each allocated part of the resource) is a crucial requirement~\cite{brams1996fair}; consider, e.g., land division, spectrum allocation, and non-preemptive interval scheduling.\footnote{The other variant of the problem, wherein agents can receive disconnected pieces, has also been studied in prior work; see, e.g., \cite{procaccia2015cake} and references therein. Related results on this variant are discussed at the end of this section.}

Achieving fairness and efficiency are two pivotal goals in this resource-allocation context~\cite{brams1996fair, robertson1998cake}. The current work contributes to these objectives, with a focus on computational aspects of cake cutting. The fairness and efficiency objectives addressed in this work are detailed next. 

A quintessential notion of fairness is envy-freeness: a division is said to be envy-free iff, under it, every agent prefers its own piece over that of any other agent~\cite{foley1967resource}. The well-known result of Su~\cite{edward1999rental} (see also \cite{stromquist1980cut} and \cite{simmons1980private}) shows that, under mild assumptions on agents' valuations, envy-free cake divisions with connected pieces always exist. However, this existential result stands without an algorithmic counterpart;  Stromquist~\cite{stromquist2008envy} has shown that such an envy-free solution cannot be computed in bounded time, if the valuations are specified by an adaptive adversary.\footnote{Notably, for the complementary problem of finding envy-free solutions with \emph{noncontiguous pieces}, the work of Aziz and Mackinzie \cite{aziz2016discrete} provides a hyper-exponential time algorithm.} This negative result leads one to study relaxations/approximation guarantees, such as the ones considered in this work.   

It is relevant to note that, while envy-freeness provides fairness guarantees on an individual level, in and of itself, this notion is not concerned with overall efficiency. By contrast, the concept of social welfare quantifies efficiency achieved by the agents as a whole. Social (utilitarian) welfare is defined as the sum of the values that the agents have for their own pieces. For this welfare objective, we establish (multiplicative) approximation guarantees in the cake-cutting setup.\footnote{Here, without loss of generality and to connect social welfare with the other objectives studied in this work, we equate social welfare with the average (arithmetic mean) of the values obtained by the agents.} 

A balance between the  utilitarian (Benthamite) objective and the egalitarian (Rawlsian/max-min) welfare is achieved through Nash social welfare ($\NSW$), which is defined as the geometric mean of the agents' values \cite{nash1950bargaining, kaneko1979nash}. This welfare function has traditionally been studied for homogeneous (divisible) goods, where it is known to possess strong fairness (envy-freeness) and (Pareto) efficiency properties \cite{varian1973equity}. The appeal of Nash social welfare continues to hold in the case of indivisible goods: Caragiannis et al. \cite{caragiannis2016unreasonable} have shown that (under additive valuations) Nash-optimal allocations of discrete goods satisfy a natural relaxation of envy-freeness and are Pareto efficient. The relevance of Nash social welfare, as a measure of fairness, motivates its study in the cake-cutting setup as well. Towards this end, this work develops algorithmic and hardness results for the problem of finding cake divisions that maximize Nash social welfare. 

Generalized (H\"{o}lder) means provide a unified framework to address fairness and efficiency objectives. Specifically, with exponent parameter $\rho$, the $\rho$-mean welfare,  $\mathrm{M}_\rho(\cdot)$, is defined as $\left( \frac{1}{n} \sum_i v_i^\rho \right)^{1/\rho}$; here $v_i$s denote the valuations obtained by the agents in an allocation. We address $\rho$-mean welfare for $\rho \in (0,1]$. In particular, this parameter range captures both Nash social welfare and social welfare: $\rho=1$ gives us the arithmetic mean (average social welfare) and, as $\rho$ tends to zero, the limit of $\mathrm{M}_\rho$ is equal to the geometric mean (the Nash social welfare).\footnote{Note that, for each $\rho$, the $\rho$-mean welfare is ordinally equivalent to CES (constant elasticity of substitution) welfare functions that have the form $\left(\sum_i v_i^\rho\right)^{1/\rho}$.} 

This paper addresses all of the above-mentioned notions of fairness and efficiency. In particular, we develop approximation algorithms for finding cake divisions, with connected pieces, under the following objectives: (i) multiplicatively bounding envy, (ii) maximizing Nash social welfare, and (iii) maximizing $\rho$-mean welfare, for $\rho \in (0,1]$. We complement these approximation guarantees by establishing hardness results for Nash social welfare and $\rho$-mean welfare maximization. Our contributions are summarized in the following list.
\begin{itemize}
\item {\bf Envy-Freeness:} We develop an efficient algorithm that finds a cake division (i.e., a partition of the cake into $n$ connected pieces along with a one-to-one assignment of these pieces among the $n$ agents) such that for every agent $a$ the value of its piece is at least $1/(2 + o(1))$ times $a$'s value for any other agent's piece (Theorem~\ref{theorem:mod2ef}). Our algorithm for finding an approximately envy-free allocation is rather direct (see Section \ref{section:ef-nsw} for a description). The explainablity/simplicity of this algorithm is a notable feature, since it makes the developed method amenable for realistic implementations, such as the ones found on websites like Spliddit~\cite{goldman2014spliddit}.

\item {\bf Nash Social Welfare:} Our algorithm for finding approximately envy-free divisions also provides a polynomial-time $\left(3 + o(1)\right)$-approximation algorithm for the Nash social welfare maximization problem (Theorem~\ref{theorem:nsw}). 

We further show that approximately envy-free cake divisions (with connected pieces) always have near-optimal Nash social welfare: if in a cake division the envy is (multiplicatively) bounded within a factor of $\alpha$, then the Nash social welfare of the division is at least $\frac{1}{2 \alpha}$ times the optimal (Theorem~\ref{theorem:ef-nsw}).\footnote{In comparison to this generic connection between envy-freeness and Nash social welfare, the cake divisions computed specifically by our algorithm admit a stronger guarantee--they essentially achieve an approximation bound of three for both envy and Nash social welfare.} 

Connections between envy-freeness and Nash social welfare have been established in other fair-division settings: addressing fair allocation of homogeneous,\footnote{Hence, such goods do not correspond to a heterogeneous cake.} divisible goods under additive valuations, the work of Varian \cite{varian1973equity} shows that there always exists an allocation which is both envy-free and Nash optimal.\footnote{In fact, in the homogenous-goods case, such an allocation can be efficiently computed by solving the convex program of Eisenberg and Gale~\cite{eisenberg1959consensus}. By contrast, finding a Nash optimal allocation is the cake-division  setting is computationally hard.} Also, Caragiannis et al.~\cite{caragiannis2016unreasonable} have established that, when dividing indivisible goods, allocations that maximize Nash social welfare satisfy relaxations of envy-freeness. Our result (Theorem~\ref{theorem:ef-nsw}) shows that analogous connections hold in the cake-division framework as well.   

We complement the algorithmic result for Nash social welfare by showing that it is {\rm APX}-hard to find a Nash optimal cake division with connected pieces (Theorem \ref{theorem:nswhardness}). This hardness result implies, in particular, that the problem of maximizing Nash social welfare does not admit a polynomial-time approximation scheme (PTAS), unless {\rm P} $=$ {\rm NP}. 

\item {\bf Generalized-Mean Welfare:} As mentioned previously, generalized means---$\mathrm{M}_\rho(\cdot)$ with exponent parameter $\rho \in (0,1]$---is a family of functions which captures both Nash social welfare and (average) social welfare. For this unified objective, we develop a $\left(2 + o(1) \right)^{1/ \rho}$-approximation algorithm that runs in time $n^{\mathcal{O}(1/\rho)}$; here $n$ is the number of agents (Theorem~\ref{theorem:interval-approx}).  Hence, for average social welfare (i.e., the $\rho=1$ case) we obtain a polynomial-time $\left(2 + o(1) \right)$-approximation algorithm. We note that this instantiation improves upon the $8$-approximation guarantee obtained specifically for social welfare in the work of Aumann et al. \cite{aumann2013computing}. 

Our algorithm, for maximizing $\rho$-mean welfare, relies on ``discretizing'' the given cake-division instance to obtain an interval-scheduling problem, called the Job Interval Selection Problem ($\JISP$). Then, we invoke the $2$-approximation algorithm of Bar-Noy et al. \cite{bar2001approximating} for $\JISP$ to obtain the stated approximation guarantee (Lemma~\ref{lemma:discretization} and Theorem~\ref{theorem:interval-approx}). 

We also establish that, for any fixed $\rho \in (0,1]$, finding cake divisions that maximize $\rho$-mean welfare is $\APX$-hard. This general result, though, holds for cake-division instances wherein the valuations are not necessarily normalized.\footnote{Agents' valuations are said to be normalized iff, for every agent, the value of the entire cake is equal to one.} For the social welfare case (i.e., the $\rho = 1$ setting), our techniques can be adopted to establish $\APX$-hardness even under normalized valuations. Hence, we can rule out a PTAS for the social welfare maximization problem. This strengthens the inapproximability result of Aumann et al.~\cite{aumann2013computing}, which showed that efficient (in the social-welfare sense) cake cutting does not admit a fully polynomial-time approximation scheme (FPTAS).

Prior work has also studied the impact of envy-freeness on social welfare in the cake-cutting context. Specifically, Caragiannis et al.~\cite{caragiannis2009efficiency} along with Aumann and Dombb~\cite{aumann2010efficiency} establish bounds for price of envy-freeness, which is defined as the ratio between the social welfare of an optimal allocation and the social welfare of the best envy-free allocation. We extend this framework to $\rho$-mean welfare and show that any (approximately) envy-free allocation provides an $\mathcal{O}(2^{\frac{1}{\rho}} \ n^{\frac{\rho}{\rho+1}})$-approximate solution to maximizing $\rho$-mean welfare, for $\rho \in (0,1]$ (Theorem~\ref{theorem:priceEF} in Section~\ref{section:priceEF}). We note that our upper bound on the price of envy-freeness for the $\rho=1$ instantiation (i.e., for social welfare) is essentially tight. This follows from considering the result of Aumann and Dombb~\cite{aumann2010efficiency}, which establishes a ${\Theta}(\sqrt{n})$ bound on the price of envy-freeness, in the social-welfare context.

\end{itemize}
\noindent 
\subsection{Additional Related Work} Another standard notion of fairness is proportionality. This criterion requires that every agent $a$ receives a piece of value at least $1/n$ times $a$'s value for the entire cake; here $n$ is the total number of agents participating in the cake-cutting exercise. In contrast to envy-freeness, proportionality is an algorithmically tractable solution concept; see~\cite{procaccia2015cake} and references therein. Though, given that an (approximately) envy-free allocation is also (approximately) proportional,\footnote{We conform to the standard assumption that the valuations of the agents over the cake are sigma additive.} approximation guarantees for envy-freeness (such as the ones developed in this work) give us matching bounds for proportionality as well.  

With respect to maximizing social welfare, the result closest to ours is that of Aumann et al.~\cite{aumann2013computing}. We reiterate that the current work improves upon the algorithmic and hardness bounds obtained in~\cite{aumann2013computing}. Bei et al. \cite{bei2012optimal} develop approximation results for maximizing social welfare with proportionality as a constraint. By contrast, we focus on social welfare by itself.  

Deng et al.~\cite{deng2012algorithmic} present an algorithm that finds an additive approximation to an envy-free cake division with connected pieces. This algorithm, however, runs in exponential (in the number of agents) time. 

If disconnected pieces can be assigned to  each agent, then an additive approximation to envy-free divisions can be computed efficiently,~see, e.g., \cite{lipton2004approximately} and the reentrant version of the last diminisher protocol in \cite{brams1996fair}. Also, for the disconnected-pieces variant and under specific valuations types, Aziz and Ye \cite{aziz2014cake} present an efficient algorithm for maximizing Nash social welfare. The results of Kurokawa et al.~\cite{kurokawa2013cut} and Cohler et al.~\cite{cohler2011optimal} address the noncontiguous-pieces setup as well. In particular, for a class of valuations, Cohler et al. \cite{cohler2011optimal} develop an algorithm for maximizing social welfare subject to the envy-freeness constraint. Our results are incomparable with all of these prior works, since we solely focus on allocation of connected pieces. 

A relaxation of envy-free division entails free disposal. The idea here is to achieve envy-freeness at the cost of discarding some parts of the cake. Aziz and Mackenzie~\cite{aziz2016discrete} (see also~\cite{segal2015waste}) develop an \emph{exponential-time} algorithm that finds envy-free divisions, wherein each agent $a$ receives a connected piece of value at least $\frac{1}{2n}$ times $a$'s value for the entire cake. Here, the proximity to the proportional share is used to quantify the loss incurred due to the disposal. Our cake-cutting algorithm ($\ALG$ in Section~\ref{section:ef-nsw}) can be used to efficiently find, with disposal, an additively-approximate envy-free division, wherein each agent receives a connected piece of value $\left(\frac{1}{2n} - o(1)\right)$ times its value for the cake.\footnote{This observation is a direct consequence of Lemma~\ref{lemma:partial-ef}.} We note that such a division is obtained as an intermediate solution in our algorithm; at the end, the algorithm allocates the entire cake.  

\subsection{Subsequent Work}

The recent result of Goldberg et al. \cite{goldberg2020contiguous} provides an efficient algorithm for computing
allocations in which the envy between any two agents is at most $1/3$. That is, they obtain an additive-approximate guarantees for contiguous envy-free cake cutting (where every agent values the entire cake at $1$). We note that our algorithm developed in Section~\ref{section:2factorEF}--- that finds a cake division (with connected pieces) wherein the envy is multiplicatively within a factor of $2 + o(1)$; see Theorem~\ref{theorem:mod2ef}---essentially matches the additive-approximation guarantee obtained in~\cite{goldberg2020contiguous}. 

The work of Goldberg et al. \cite{goldberg2020contiguous} also prove interesting hardness results in the cake cutting context; in particular, they show that the decision problem of whether there exists a contiguous envy-free allocation satisfying the following constraints is NP-hard: (i) a certain agent must be allocated the leftmost piece; (ii) the ordering of the agents is fixed; or (iii) one of the cuts must fall at a given position.





\section{Notation and Preliminaries}

We consider the problem of dividing a cake (which metaphorically represents a divisible, heterogenous good) among $n$ agents. In this setup, the cake is modeled as the segment $[0,1]$ and the (possibly) distinct cardinal preferences of the agents are expressed as valuation functions, $\{ v_a \}_{a \in [n]}$, over the intervals contained in $[0,1]$ (i.e., over the pieces of the cake). Specifically, for each agent $a \in [n]$ and interval $I =[x,y] \subset [0,1]$, with $0 \leq x \leq y \leq 1$, the function $v_a$ maps $I$ to agent $a$'s value for it, $v_a (I) \in \mathbb{R}_+$. 

Conforming to standard assumptions, this work addresses valuations $\{v_a\}_{a \in [n]}$ that are (i) nonnegative, (ii) normalized: the value of the entire cake is equal to one, $v_a([0,1]) =1$, (iii) divisible: for every interval $I = [x,y]$ and parameter $ \lambda \in [0,1]$, there exists a $z \in [x,y]$ with the property that $v_a([x,z]) = \lambda v_a([x,y])$, and (iv) sigma additive: $v_a(I \cup J) = v_a(I) + v_a(J)$, for all disjoint intervals $I, J \subset [0,1]$.  

This divisibility property ensures that the valuations are non-atomic, i.e., $v_a([x,x]) = 0$ for all $a \in [n]$ and $x \in [0,1]$. Furthermore, this property allows us, as a convention, to regard two intervals to be disjoint even if they intersect exactly at an endpoint.  	

Our results hold as long as the valuations satisfy the above-mentioned properties and only require oracle access to the valuations. That is, our algorithms can be efficiently executed in the Robertson-Webb model~\cite{robertson1998cake}, which supports oracle access to the valuations in the form of evaluation queries (which, given an agent $a$ and an interval $I$, return $v_a(I)$) and cut queries (which, given an agent $a$, an initial point $x \in [0,1]$, and value $\tau$, return the leftmost point $y \in [x,1]$ such that $v_a([x,y]) = \tau$). 

However, for ease of presentation, instead of the Robertson-Webb model, we will restrict attention to a well-studied setting in which the valuations of the agents can be explicitly given as input. In particular, we will consider valuations that are induced by density functions: given a piecewise-constant density function $\nu_a: [0,1] \mapsto \mathbb{R}_+$ for an agent $a \in [n]$, the valuation of any interval $I$ is set to be $v_a(I) \coloneqq \int_I \! \nu_a(x) \, \mathrm{d}x$. Valuations obtained by integrating piecewise-constant densities are said to be \emph{piecewise-constant}. Indeed, such valuations can be given as input, say, in terms of the underlying density functions. \\

\noindent
{\bf Problem Instances:} A \emph{cake-division instance}, with piecewise-constant valuations, is a tuple $\langle [n], \{v_a \}_{a \in [n] } \rangle$ where $[n]=\{1,2, \ldots, n\}$ denotes the set of $n$ agents and $v_a$s specify the piecewise-constant valuations of the agents over the cake $[0,1]$. \\

\noindent
{\bf Allocations:} As mentioned above, the goal here is to partition the cake into disjoint intervals and allocate them among the $n$ agents. We will focus solely on assigning to each agent a single interval, i.e., we will require that the piece assigned to each agent is connected.

For a cake-division instance with $n$ agents, an \emph{allocation} is defined to be a collection of $n$ pairwise-disjoint intervals, $\mathcal{I} = \{I_1, I_2, \ldots, I_n \}$, where interval $I_a$ is assigned to agent $a \in [n]$ and $\cup_{a \in [n]} \ I_a = [0,1]$.\footnote{Note that the intervals are not indexed based on how their endpoints are ordered, rather the subscript of each interval in an allocation identifies the unique agent that owns this interval.} We will use the term \emph{partial allocation} to refer to collection of pairwise-disjoint intervals, $\mathcal{J} = \{J_1, J_2, \ldots, J_n \}$, that do not necessarily cover the entire cake, $\cup_a J_a \subsetneq [0,1]$. 

The overarching objective of the current work is to find fair and efficient allocations. Relevant notions of fairness and efficiency are defined next. \\

\noindent
{\bf Envy-Freeness:} For a cake-division instance $\langle [n], \{v_a\}_{a \in [n]} \rangle$, an (partial) allocation $\mathcal{I} = \{I_1, \ldots, I_n \}$ is said to be \emph{envy free} ($\EF$) iff each agent prefers its own interval over that of any other agent, $v_a(I_a) \geq v_a(I_b)$ for all agents $a, b \in [n]$. 

We will address a natural relaxation of envy-freeness; specifically, we study allocations in which the envy between the agents is multiplicatively bounded. Given $\alpha \geq 1$, an allocation $\mathcal{I}=\{I_1, \ldots, I_n\}$ is said to be \emph{$\alpha$-approximately envy free} ($\alpha$-$\EF$) iff $v_a( I_a) \geq \frac{1}{\alpha} \ v_a (I_b)$, for every pair of agents $a, b \in [n]$.

A $1$-$\EF$ allocation is envy free and, the smaller the value of $\alpha$, the stronger is the envy-freeness guarantee.\footnote{Also, note that $\alpha$ cannot be strictly less than one--the definition of an $\alpha$-$\EF$ allocation requires $v_a( I_a) \geq \frac{1}{\alpha} \ v_a (I_b)$, even for $b=a$.} \\

\noindent
{\bf Nash Social Welfare:} For an allocation $\mathcal{I}=\{I_1, \ldots, I_n\}$, the \emph{Nash social welfare} is defined to be the geometric mean of the agents' valuations, $\NSW(\mathcal{I}) \coloneqq \left( \prod_{a=1}^n v_a(I_a) \right)^{1/n}$. In a cake-division instance, an allocation $\mathcal{I}^*$ is said to be a \emph{Nash optimal allocation} iff $\mathcal{I}^* \in \argmax_{ \mathcal{I} \in \mathbb{I} }  \NSW(\mathcal{I})$; here $\mathbb{I}$ denotes the set of all allocations. \\

\noindent 
{\bf Social Welfare and Generalized Mean:} Social welfare is a standard measure of efficiency in the context of resource allocation. For an allocation $\mathcal{I} = \{I_1, \ldots, I_n\}$, we define \emph{social welfare} to be the arithmetic mean\footnote{Since the work develops multiplicative approximation guarantees, we can consider the average valuation, instead of the sum of valuations, as a utilitarian objective.} of the valuations, $\SW(\mathcal{I}) \coloneqq \frac{1}{n} \sum_{a=1}^n v_a(I_a)$.

Generalized (H\"{o}lder) means, $\mathrm{M}_\rho$, provide a family of functions which interpolate between fairness and efficiency objectives. The \emph{$\rho$-mean welfare} of an allocation $\mathcal{I} = \{I_1, \ldots, I_n \}$ is defined as
\begin{align*}
\mathrm{M}_\rho(\mathcal{I}) \coloneqq \left( \frac{1}{n} \sum_{a=1}^n [v_a(I_a)]^\rho \right)^{1/\rho}
\end{align*} 

We will develop algorithmic and hardness results for maximizing the $\rho$-mean welfare, with exponent $\rho \in (0,1]$. This parameter range, in particular, captures both $\NSW$ and $\SW$: $\rho=1$ gives us the arithmetic mean (social welfare) and, as $\rho$ tends to zero, the limit of $\mathrm{M}_\rho$ is equal to the geometric mean (the Nash social welfare).

Overall, this paper is concerned with finding allocations (i.e., finding cake divisions with connected pieces) under the following objectives (i) bounding envy, (ii) maximizing Nash social welfare, and (iii) maximizing $\rho$-mean welfare, for $\rho \in (0,1]$.

\section{Finding Envy-Free and Nash Optimal Allocations}
\label{section:ef-nsw}

In this section, first we will develop an efficient algorithm for finding $\left(3 +o(1) \right)$-$\EF$ allocations and, in tandem, obtain a polynomial-time $\left(3 +o(1)\right)$-approximation algorithm for the Nash social welfare maximization problem. Next, Theorem~\ref{theorem:mod2ef} strengthens the approximation guarantees obtained in Theorem~\ref{theorem:2ef} and we develop an efficient algorithm that outputs $\left(2 +o(1) \right)$-$\EF$ allocations. Subsequently, we will establish a generic connection between envy-freeness and Nash social welfare in the cake-cutting context: any $\alpha$-approximately envy-free allocation provides a $2 \alpha$-approximation to Nash social welfare.

Our algorithm, $\ALG$, for finding approximately envy-free allocations starts by assigning an empty interval to each agent--it starts with the partial allocation consisting of empty sets. Then, the algorithm proceeds to assign successively higher valued pieces to the agents, i.e., it iteratively moves from one partial allocation to the next. Note that the initial partial allocation (consisting of empty intervals) is envy free. In fact, all the partial allocations, $\mathcal{P} = \{P_1, \ldots, P_n \}$, computed during $\ALG$'s execution, satisfy the following additive relaxation of envy-freeness, for a fixed constant $\varepsilon \in (0,1/3]$:
 \begin{align}
 v_a(P_a) & \geq v_a(P_b) - \frac{\varepsilon}{n^2} \quad \text{for all } a, b \in [n] \label{eq:invariant}
 \end{align}  
$\ALG$ updates a partial allocation $\mathcal{P} = \{P_1, \ldots, P_n \}$ by considering the unassigned pieces of the cake. Specifically, given a partial allocation $\mathcal{P}$, write $\mathcal{U}_\mathcal{P} =\{U_1, U_2, \ldots, U_m \}$ to denote the minimum-cardinality collection of disjoint intervals that satisfy $\cup_i U_i = [0,1] \setminus \cup_a P_a$. In other words, $\mathcal{U}_\mathcal{P}$ consists of the intervals that remain after the assigned intervals in $\mathcal{P}$ (i.e., $P_a$s) are removed from $[0,1]$. Since there are $n$ intervals in $\mathcal{P}$, there can be at most $n+1$ intervals in $\mathcal{U}_{\mathcal{P}}$. 

$\ALG$ keeps iterating as long as there exists an unassigned interval $\widehat{U} \in \mathcal{U}_\mathcal{P}$ of high enough value for any agent. Then, part of $\widehat{U}$ is assigned to a judiciously-chosen agent $\widehat{a}$ who relinquishes the previous interval assigned to it, but now accrues a higher valuation. The criterion for selecting $\widehat{a}$ ensures that the above-mentioned invariant is maintained; this selection can be viewed as a \emph{moving-knife procedure} applied within $\widehat{U}$ (see Figure~\ref{figure:knife}). 

\begin{figure}[h]
\begin{center}
\includegraphics[scale=.8]{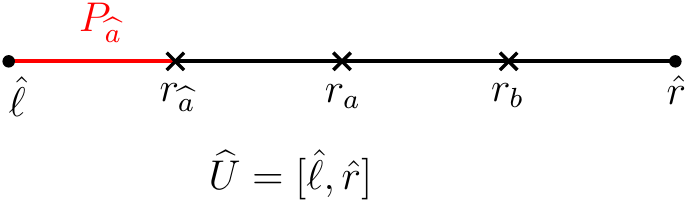}
\end{center}
\vspace*{-17pt}
\caption{An illustration of Step~\ref{step:knife} in $\ALG$}
\label{figure:knife}
\end{figure}

At the end, when the values of the unassigned intervals are not much larger than the value of the assigned ones, $\ALG$ merges each unassigned interval in $\mathcal{U}_\mathcal{P}$ with an adjacent interval in the (final) partial allocation $\mathcal{P}$ to obtain an approximately envy-free allocation (see Figure~\ref{figure:merge}). The algorithm is detailed below and we prove in Theorem~\ref{theorem:2ef} that it efficiently finds a $\left(3 + o(1) \right)$-$\EF$ allocation.
  
{
	\begin{algorithm}[h]
		{ 
				{\bf Input:} A cake-division instance $\langle [n], \{v_a\}_a \rangle$ with piecewise-constant valuations and a fixed constant $\varepsilon \in (0,1/3]$. \\ 
				{\bf Output:} A $\left(3+ \frac{9 \varepsilon}{n} \right)$-approximately envy-free allocation.
		 \caption{$\ALG$}	
			\label{alg:2ef}
			\begin{algorithmic}[1]
				\STATE Initialize partial allocation $\mathcal{P}=\{P_1, \ldots, P_n \}$ with empty intervals, i.e., $P_a = \emptyset$ for all $a \in [n]$. 
				\COMMENT{Recall that $\mathcal{U}_{\mathcal{P}}$ denotes the set of unassigned intervals induced by any partial allocation $\mathcal{P}$.}	
				\WHILE{there exists an agent $a \in [n]$ and an unassigned interval $\widehat{U} =[\hat{\ell}, \hat{r}] \in \mathcal{U}_{\mathcal{P}}$ such that $v_a(P_a) < v_a(\widehat{U}) - \frac{\varepsilon}{n^2}$} \label{step:while-loop}
				\STATE \label{step:contenders} Let $C \coloneqq \left\{ b \in [n] \ : \ v_b(P_b) < v_b(\widehat{U}) - \frac{\varepsilon}{n^2} \right\}$ and, for every agent $b \in C$, set $r_b \in [\hat{\ell}, \hat{r}] $ to be the leftmost point such that $v_b([\hat{\ell}, r_b]) = v_b(P_b) + \frac{\varepsilon}{n^2}$.
				\STATE \label{step:knife} Select agent $\widehat{a} \in \argmin_{b \in C} \  r_b$. 
				 \STATE \label{step:update} Update the partial allocation $\mathcal{P}$ by assigning $P_{\widehat{a}} \leftarrow [\hat{\ell}, r_{\widehat{a}}]$ and keeping the interval assignment of all other agents unchanged.
				\STATE Update $\mathcal{U}_{\mathcal{P}}$ to be the set of unassigned intervals induced by the current partial allocation $\mathcal{P}$.
				\ENDWHILE 
				
				\STATE \label{step:merge1} Associate each unassigned interval $U \in \mathcal{U}_{\mathcal{P}}$ with an assigned interval $P_a \in \mathcal{P}$ which is adjacent (either on the left or on the right) to $U$. \\ \COMMENT{Note that any $P_a \in \mathcal{P}$ gets associated with at most two unassigned intervals, say $U$ and $U'$, and $U \cup P_a \cup U'$ is itself an interval} 
				\STATE \label{step:merge2} For all $a \in [n]$, let interval $I_a$ be the union of $P_a$ and the unassigned intervals (if any) associated with it. 
				\RETURN allocation $\mathcal{I} = \{I_1, \ldots, I_n \}$				
			\end{algorithmic}
		}
	\end{algorithm}
}		

\begin{figure}[H]
\begin{center}
\includegraphics[scale=.65]{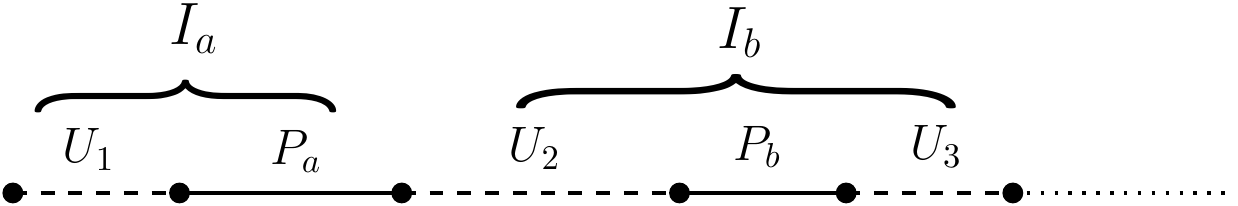}
\end{center}
\vspace*{-10pt}
\caption{An illustration of Steps~\ref{step:merge1} and~\ref{step:merge2} in $\ALG$}
\label{figure:merge} 
\end{figure}

The following lemma shows that the final partial allocation considered by $\ALG$ (in Step~\ref{step:merge1}) satisfies the additive relaxation of envy-freeness considered in the equation \eqref{eq:invariant}, not only between the assigned intervals, but also against the unassigned ones. 

\begin{lemma}
For a given a cake-division instance $\left\langle [n], \{v_a\}_{a \in [n]} \right\rangle$, with piecewise-constant valuations, and given parameter $\varepsilon \in (0, 1]$, let $\mathcal{P}=\{P_1, \ldots, P_n\}$ be the final partial allocation considered by $\ALG$ (i.e., $\mathcal{P}$ is the partial allocation with which the while-loop terminates) and let $\mathcal{U}_{\mathcal{P}}$ be the set of unassigned intervals induced by $\mathcal{P}$. Then, for each agent $a \in [n]$ we have 
\begin{align}
v_a(P_a) \geq v_a(Q) - \frac{\varepsilon}{n^2} \qquad \text{for all } Q \in \mathcal{P} \cup \mathcal{U}_{\mathcal{P}} \label{ineq:prop}
\end{align}
\label{lemma:partial-ef}
\end{lemma}
\begin{proof} 
First, we observe that the collection of intervals $\mathcal{P} =\{P_1, \ldots, P_n \}$ considered in any iteration of $\ALG$ is indeed a partial allocation, i.e., the intervals $P_a$s are pairwise disjoint. This simply follows from the fact that whenever an agent $\widehat{a}$ is assigned a new interval $P_{\widehat{a}}$, it is contained in the unassigned part of the cake, i.e., $P_{\widehat{a}}$ continues to not intersect with the intervals assigned to other agents. 

Now, we will show, via an inductive argument, that the partial allocations $\mathcal{P} =\{P_1, \ldots, P_n \}$ computed by $\ALG$ satisfy the above-mentioned invariant \eqref{eq:invariant}. The initial (empty) partial allocation is envy free and, hence, trivially satisfies this additive relaxation of envy freeness. In an iteration, only the interval assigned to the selected agent $\widehat{a}$ gets updated. Since $\widehat{a}$'s valuation increases (by $\frac{\varepsilon}{n^2}$), equation \eqref{eq:invariant} continues to hold for $\widehat{a}$, i.e., $v_{\widehat{a}}(P_{\widehat{a}}) \geq v_{\widehat{a}} ( P_b )  - \frac{\varepsilon}{n^2}$ for all $b \in [n]$. Also, for all agents $a, b \in [n] \setminus \{ \widehat{a} \}$, the assigned intervals remain unchanged, i.e., equation \eqref{eq:invariant} is satisfied for these agents. 

It only remains to verify that $v_a(P_a) \geq v_a(P_{\widehat{a}}) - \frac{\varepsilon}{n^2}$ for all $a \in [n]$. For agents $a \not\in C$ (see Step~\ref{step:contenders} for the definition of $C$), we have $v_a(P_a) \geq v_a(\widehat{U}) - \frac{\varepsilon}{n^2} \geq v_a([\hat{\ell}, r_{\widehat{a}} ]) - \frac{\varepsilon}{n^2}$; the last inequality follows from the containment $[\hat{\ell}, r_{\widehat{a}} ]  \subseteq \widehat{U}$. Since $P_{\widehat{a}}$ is set to be $[\hat{\ell}, r_{\widehat{a}} ] $, the desired inequality holds for all $a \not\in C$. For the remaining agents $a \in C$, the selection criterion applied in Step~\ref{step:knife} gives us this inequality. 
Indeed, $r_{\widehat{a}} \leq r_a$ for all $a \in C$. Hence, for agents $a \in C$, we have $v_a(P_a) = v_a([\hat{\ell}, r_a]) - \frac{\varepsilon}{n^2} \geq v_a([\hat{\ell}, r_{\widehat{a}}]) - \frac{\varepsilon}{n^2} = v_a(P_{\widehat{a}}) - \frac{\varepsilon}{n^2}$. 

Therefore, invariant \eqref{eq:invariant} holds, in particular, for the partial allocation $\mathcal{P}$ obtained at the termination of the while-loop. Also, the fact that the while-loop terminates with $\mathcal{P}$ as the final partial allocation implies that, for all $a \in [n]$ and all $U \in \mathcal{U}_{\mathcal{P}}$, we have $v_a(P_a) \geq v_a(U) - \frac{\varepsilon}{n^2}$. That is, for each $a \in [n]$, the final partial allocation $\mathcal{P}= \{P_1, \ldots, P_n\}$ satisfies the stated inequalities: $v_a(P_a) \geq v_a(Q) - \frac{\varepsilon}{n^2}$ for all  $Q \in \mathcal{P} \cup \mathcal{U}_{\mathcal{P}}$. 
\end{proof}

Using this lemma, we will now show that the allocation computed by $\ALG$ is $\left(3+o(1) \right)$-$\EF$.

\begin{restatable}{theorem}{TheoremEF}
\label{theorem:2ef}
Given a cake-division instance $\left\langle [n], \{v_a\}_{a \in [n]} \right\rangle$ with piecewise-constant valuations, and constant $\varepsilon \in (0, 1/3]$, $\ALG$ computes a $(3+ \frac{9 \varepsilon}{n})$-approximately envy-free allocation in polynomial time.
\end{restatable}
	
\begin{proof}
To bound the algorithm's time complexity note that in every iteration the selected agent's valuation (for the interval assigned to it) additively goes up by $\frac{\varepsilon}{n^2}$: in Steps \ref{step:contenders} and \ref{step:knife}, for the selected agent $\widehat{a}$, we have $v_{\widehat{a}} ( [\hat{\ell}, r_{\widehat{a}}]) = v_{\widehat{a}} (P_{\widehat{a}}) + \frac{\varepsilon}{n^2}$. Since the total value of the cake for every agent is equal to one, $\ALG$ will iterate at most $\varepsilon^{-1}\ n^3$ times. Also, note that every step of the algorithm can be implemented efficiently and $\varepsilon$ is set to be a constant. Hence, $\ALG$ runs in polynomial time.

As observed in the proof of Lemma~\ref{lemma:partial-ef}, the collection of intervals, $\mathcal{P} =\{ P_1, \ldots, P_n \}$, considered by $\ALG$ in Step~\ref{step:merge1} is indeed a partial allocation, i.e., the intervals $P_a$s with which the while-loop terminates are pairwise disjoint. Let $\mathcal{U}_{\mathcal{P}}$ be the set of unassigned intervals induced by the final partial allocation $\mathcal{P}$. Also, write $\mathcal{I}=\{I_1, \ldots, I_n\}$ to denote the allocation returned by $\ALG$; note that $P_a \subseteq I_a$ for all agents $a \in [n]$. Also, since $\mathcal{P}$ contains $n$ intervals,  $|\mathcal{U}_{\mathcal{P}}| \leq n+1$. 

Summing inequality \eqref{ineq:prop} (see Lemma~\ref{lemma:partial-ef}) across all intervals $Q \in \mathcal{P} \cup \mathcal{U}_{\mathcal{P}}$ gives us  
\begin{align}
(2n +1) \ v_a(P_a)  \geq  \sum_{Q \in \mathcal{P} \cup \ \mathcal{U}_{\mathcal{P}} } v_a(Q) - 2n \frac{\varepsilon}{n^2} = 1 - \frac{2\varepsilon}{n} \label{ineq:half-prop}
\end{align}  
The last equality holds since $\bigcup_{Q \in \mathcal{P} \cup \ \mathcal{U}_{\mathcal{P}}} \ Q  = [0,1]$.

This inequality provides the following lower bound on the value attained by any agent $a \in [n]$ in the returned allocation $\mathcal{I} = \{I_1, \ldots, I_n\}$: $v_a(I_a) \geq v_a(P_a) \geq \frac{1}{2n+1} - \frac{2 \varepsilon}{n(2n+1)}$. Therefore, with $n \geq 3$ and $\varepsilon \leq 1/3$,\footnote{For the $n= 2$ case one can efficiently find an envy-free allocation (i.e., a $1$-$\EF$ allocation) by the cut-and-choose protocol~\cite{procaccia2015cake}.} we have the following bound: 
\begin{align}
\frac{3\varepsilon}{n} v_a(I_a) \geq \frac{\varepsilon}{n^2} \label{ineq:half-prop-use}
\end{align}

By construction, for each agent $b \in [n]$, the returned interval $I_b$ is composed of $P_b$ and at most two other unassigned intervals from $\mathcal{U}_{\mathcal{P}}$. Therefore, instantiating inequality \eqref{ineq:prop} with $P_b \in \mathcal{P} \cup \mathcal{U}_{\mathcal{P}}$ and the (at most two) unassigned intervals associated with it, we get $3 v_a(P_a) \geq v_a(I_b) - \frac{3\varepsilon}{n^2}$. That is, $3 v_a(I_a) \geq v_a(I_b) - \frac{3\varepsilon}{n^2}$.

Using this inequality and the bound  \eqref{ineq:half-prop-use}, we obtain the desired approximate envy-freeness guarantee  $\left(3 + \frac{9\varepsilon}{n} \right) v_a(I_a) \geq v_a(I_b)$ for all $a, b \in [n]$.
\end{proof}

\begin{remark2}
Theorem~\ref{theorem:2ef} provides a proof of existence of approximately envy-free cake divisions. Indeed, this existential guarantee also follows from the (stronger) result of Su~\cite{edward1999rental}. However, in contrast to~\cite{edward1999rental}, the current proof renders an efficient algorithm and relies on a potential argument--the proof in~\cite{edward1999rental} invokes Sperner's Lemma and, hence, achieves totality through a parity argument. 
\end{remark2}

\subsection{Finding $\left(2+o(1) \right)$-$\EF$ Allocations:} \label{section:2factorEF}
In this section, we develop a new algorithm that builds on $\ALG1$ and improves the approximation guarantee obtained in Theorem~\ref{theorem:2ef}. Note that if the partial allocation $\mathcal{P} = \{P_1, P_2, \dots, P_n\}$ considered in Step~\ref{step:merge1} of $\ALG1$ induces at most $n$ unassigned intervals, we achieve the desired $\left(2+o(1) \right)$-approximation guarantee and prove Theorem~\ref{theorem:mod2ef}. We therefore begin by identifying various properties of a partial allocation that induces no more than $n$ unassgined intervals, and detail required additional steps in $\ALG1$ that ensures at least one of these properties are maintained throughout its execution.

Write $\mathcal{U}(\mathcal{P})$ to denote the set of unassigned intervals induced by the partial allocation $\mathcal{P}$. Since $\mathcal{P}$ contains $n$ intervals,  $|\mathcal{U}_{\mathcal{P}}| \leq n+1$. Note that $\mathcal{P}$ can induce $n+1$ unassigned intervals iff there are unassigned intervals at both ends of the cake $[0,1]$ \emph{and} no two assigned intervals are adjacent to each other. Therefore, $|\mathcal{U}_{\mathcal{P}}| \leq n$ iff the partial allocation $\mathcal{P}$ satisfies either of the following two properties: 
\begin{enumerate}
	\item There exists some assigned interval $P_a \in \mathcal{P}$, for $a\in[n]$, that is adjacent to either ends, $0$ or $1$, of the cake. Here, we can denote $\mathcal{P}$  as  $|P_a \ U_a \ \cdots|$ or $| \cdots \ U_a \ P_a|$.
	\item There exists at least two assigned intervals $P_b, P_c \in \mathcal{P}$, for $b,c \in[n]$, that are adjacent to each other. Here, we can denote $\mathcal{P}$  as $|\cdots \ U_b \ P_b \ P_c \ U_d \ P_d \ \cdots|$.
\end{enumerate}

The key idea is to maintain a partial allocation that satisfies at least one of the above-mentioned two properties throughout the execution of $\ALG1$.
Towards this, we detail the additional steps required in $\ALG1$ that enables us to achieve this goal. During the execution of the while-loop in $\ALG1$, we identify cases where an application of Steps~\ref{step:contenders} and \ref{step:knife} on the unassigned interval $\hat{U}$ (selected in Step~\ref{step:while-loop}) would  violate properties (i) or (ii), and in such settings we execute a \emph{moving-knife procedure} from the 
right end of $\hat{U}$; specifically, we execute the following steps instead:

\noindent
\textit{Step 3':} Let $C \coloneqq \left\{ b \in [n] \ : \ v_b(P_b) < v_b(\widehat{U}) - \frac{\varepsilon}{n^2} \right\}$ and, for every agent $b \in C$, set $\ell_b \in [\hat{\ell}, \hat{r}] $ to be the right-most point such that $v_b([\ell_b, \hat{r}]) = v_b(P_b) + \frac{\varepsilon}{n^2}$. 

\noindent
\textit{Step 4':} Select agent $\widehat{a} \in \argmax_{b \in C} \  \ell_b$.

The next claim proves that addition of the above steps to $\ALG1$ ensures that no more than $n$ unassigned intervals can be induced by any partial allocation considered throughout the execution of $\ALG1$.

\begin{claim} \label{claim:unassigned-intervals}
	For a given cake-division instance $\left\langle [n], \{v_a\}_{a \in [n]} \right\rangle$, with piecewise-constant valuations, and a parameter $\varepsilon \in (0, 1]$, the above-mentioned additional steps to the while-loop in $\ALG1$ ensures that the final partial allocation $\mathcal{P}=\{P_1, \ldots, P_n\}$ considered in its Step~\ref{step:merge1} induces at most $n$ unassigned intervals.
\end{claim} 
\begin{proof}
	During the execution of the while-loop in $\ALG1$, let us consider the first iteration where the partial allocation $\mathcal{P}$ induces exactly $n$ unassigned intervals or $|\mathcal{U}_{\mathcal{P}}| = n$, i.e., it satisfies at least one of the above-mentioned properties. Our goal then is to maintain $|\mathcal{U}_{\mathcal{P}}| \leq n$ till the completion of the while-loop. Write $\widehat{U}$ to denote the unassigned interval considered during the next iteration of while-loop in $\ALG1$ (Step~\ref{step:while-loop}).
	We inspect the two properties described in the beginning of Section~\ref{section:2factorEF} and identify the settings where an application of Steps~3' and 4' (instead of Steps~\ref{step:contenders} and \ref{step:knife}) in the while-loop maintain $|\mathcal{U}_{\mathcal{P}}| \leq n$.\\
	
	\noindent
	\textbf{Case(i)}: The partial allocation $\mathcal{P}$ in the current execution of the while-loop of $\ALG1$ satisfies property~(1).

	\noindent
	If the partial allocation, $\mathcal{P}$ is of the form $| \cdots \ U_a \ P_a|$ with $P_a = [x,1]$ for some $x\in[0,1]$ and $a \in [n]$, then  Steps~\ref{step:contenders} and \ref{step:knife}
	ensures $|\mathcal{U}_{\mathcal{P}}| \leq n$ in the next iteration. This is due to the fact that the partial allocation continue to satisfy either property (i) or (ii). Therefore, we consider the other case when $\mathcal{P}$ is of the form $|P_a \ U_a \ \cdots|$ with $P_a = [0,x]$ for some $x\in[0,1]$ and $a \in [n]$.
	
	\noindent
	\textit{Sub-case 1:} The unassigned interval, $\widehat{U}$ considered during the next iteration of while-loop (Step~\ref{step:while-loop} of $\ALG1$) is \emph{not} $U_a$. 
	
	Here, an application of Steps~\ref{step:contenders} and \ref{step:knife} on $\widehat{U}$
	ensures that the partial allocation (after this iteration) acquires property (ii), and hence maintains $|\mathcal{U}_{\mathcal{P}}| \leq n$.
	
	\noindent
	\textit{Sub-case 2:} The unassigned interval, $\widehat{U}$ considered during the next iteration of while-loop (Step~\ref{step:while-loop} of $\ALG1$) is $U_a$. 
	
	In this case, we apply Steps 3' and 4' in the while-loop (instead of Steps~\ref{step:contenders} and \ref{step:knife})  and perform a \emph{moving-knife procedure} within $\widehat{U}=U_a$ from the right end. If agent $\widehat{a}$ selected in Step 4' is such that $\widehat{a} \neq a$, then 
	the partial allocation (after this iteration) retains property (i). Otherwise, it acquires property (ii).
	In other words, if $\widehat{a}=a$, $P_a$ becomes adjacent to some other assigned interval, and the partial allocation acquires property (ii). Otherwise, $P_a$ remains unchanged and then property (i) is retained. Overall, 
	$|\mathcal{U}_{\mathcal{P}}| \leq n$ is maintained through the next iteration.\\
	
	\noindent
	\textbf{Case(ii)}: The partial allocation $\mathcal{P}$ in the current execution of the while-loop of $\ALG1$ satisfies property~(2).
	
	\noindent
	In this case, we denote $\mathcal{P}$ as $|\cdots \ U_b \ P_b \ P_c \ U_d \ P_d \ \cdots|$ where $P_b$ and $P_c$ are two assigned intervals that are adjacent to each other.
	
	\noindent
	\textit{Sub-case 1:} The unassigned interval, $\widehat{U}$ considered during the next iteration of while-loop (Step~\ref{step:while-loop} of $\ALG1$) is \emph{not} $U_d$. 
	
	Here, an application of Steps~\ref{step:contenders} and \ref{step:knife} on $\widehat{U}$
	ensures that the partial allocation (after this iteration) either retains property (ii), or acquires property (i) if $\widehat{U}$ is adjacent to the left end of the cake $[0,1]$. Overall, 
	$|\mathcal{U}_{\mathcal{P}}| \leq n$ is maintained through the next iteration.

	\noindent
	\textit{Sub-case 2:} The unassigned interval, $\widehat{U}$ considered during the next iteration of while-loop (Step~\ref{step:while-loop} of $\ALG1$) is $U_d$. 
	
	In this case, we apply Steps 3' and 4' in the while-loop (instead of Steps~\ref{step:contenders} and \ref{step:knife}) and perform a \emph{moving-knife procedure} within $\widehat{U}=U_d$ from the right end. If agent $\widehat{a}$ selected in Step~4' is agent $b$ (or agent $c$) then it ensures that $P_b$ (or $P_c$) becomes adjacent to $P_d$. Otherwise, both $P_b$ and $P_c$ remain unchanged and continue to be adjacent to each other. That is, the above additional steps guarantee that the partial allocation (after this iteration) retains property (i).

	Overall, the addition of Steps~3' and 4' (in the while-loop) to $\ALG1$ ensures that the partial allocation considered in its Step~\ref{step:merge1} induces at most $n$ unassigned intervals, and hence completes the proof.
\end{proof}

Claim~\ref{claim:unassigned-intervals} therefore ensures that $\ALG1$ with the above-mentioned additions (Steps~3' and 4') maintains a stronger bound: each partial allocation $\mathcal{P}$ (and in particular, the one considered in its Step~\ref{step:merge1}) computed by the algorithm induces at most $n$ unassigned intervals, i.e., $|\mathcal{U}_\mathcal{P}|  \leq n$. Claim~\ref{claim:unassigned-intervals} combined with the proof of Theorem~\ref{theorem:2ef} therefore proves the following result.

\begin{restatable}{theorem}{TheoremEFm} \label{theorem:mod2ef}
	Given a cake-division instance $\left\langle [n], \{v_a\}_{a \in [n]} \right\rangle$ with piecewise-constant valuations, and constant $\varepsilon \in (0, 1/3]$, there exists a polynomial-time algorithm that computes a $(2+ \frac{9 \varepsilon}{n})$-approximately envy-free allocation in polynomial time.
\end{restatable}


\subsection{Finding Nash Optimal Allocations}
Next, we will show that the allocations computed by $\ALG$ are not only $\left(3 + o(1) \right)$-$\EF$, but they also provide a $\left(3 + o(1) \right)$-approximation to Nash social welfare. 

The following theorem shows that an approximation ratio close to $3$ can be obtained for the Nash social welfare maximization problem when the number of agents, $n$, is appropriately large. Such an approximation guarantee can also be achieved for constant values of $n$. This follows from the observation that, for the Nash social welfare maximization problem, one can compute an $\alpha$-approximate solution (with $\alpha >1$) in time $\left(\frac{n}{\log \alpha} \right)^{\mathcal{O} \left( n \right)}$; see Appendix~\ref{appendix:const-num-agents-nsw} for details. Therefore, for any number of agents, maximizing Nash social welfare admits a polynomial-time $\left( 3 + o(1) \right)$-approximation algorithm.  


\begin{restatable}{theorem}{TheoremNSW}
\label{theorem:nsw}
In cake-division instances with piecewise-constant valuations, the problem of maximizing Nash social welfare (with connected pieces) admits a polynomial-time $\left(3+ \frac{5}{n}\right)$-approximation algorithm; here $n$ is the number of agents participating in the cake-cutting exercise.  	
\end{restatable}
\begin{proof}
Let $\mathcal{I}=\{I_1, \ldots, I_n\}$ be the allocation returned by $\ALG$ and  write $\mathcal{P}=\{P_1, \ldots, P_n\}$ to denote the final partial allocation considered by $\ALG$ (i.e., $\mathcal{P}$ is the partial allocation with which the while-loop terminates). As before, $\mathcal{U}_{\mathcal{P}}$ denotes the set of unassigned intervals induced by $\mathcal{P}$. Note that, for each agent $a \in [n]$, the following containment holds $P_a \subset I_a$; hence, $v_a(I_a) \geq v_a (P_a)$. 

Let $I^*_a$ denote the interval assigned to agent $a \in [n]$ in the Nash optimal allocation $\mathcal{I}^* = \{I^*_1, \ldots, I^*_n\}$. Also, write $K_a$ to denote the set of intervals in the collection $\mathcal{P} \cup \mathcal{U}_{\mathcal{P}}$ that intersect with $I^*_a$ (see Figure~\ref{figure:overlap} for an illustration), $K_a \coloneqq \{ Q  \in \mathcal{P}  \cup \mathcal{U}_{\mathcal{P}} \mid Q \cap I^*_a \neq \emptyset \}$.\footnote{Here, we follow the above-mentioned convention that mandates two intervals to be disjoint, if they intersect exactly at an endpoint.} Let $k_a$ denote the cardinality of this set, $k_a \coloneqq |K_a|$.   
\begin{figure}[h]
\begin{center}
\includegraphics[scale=.62]{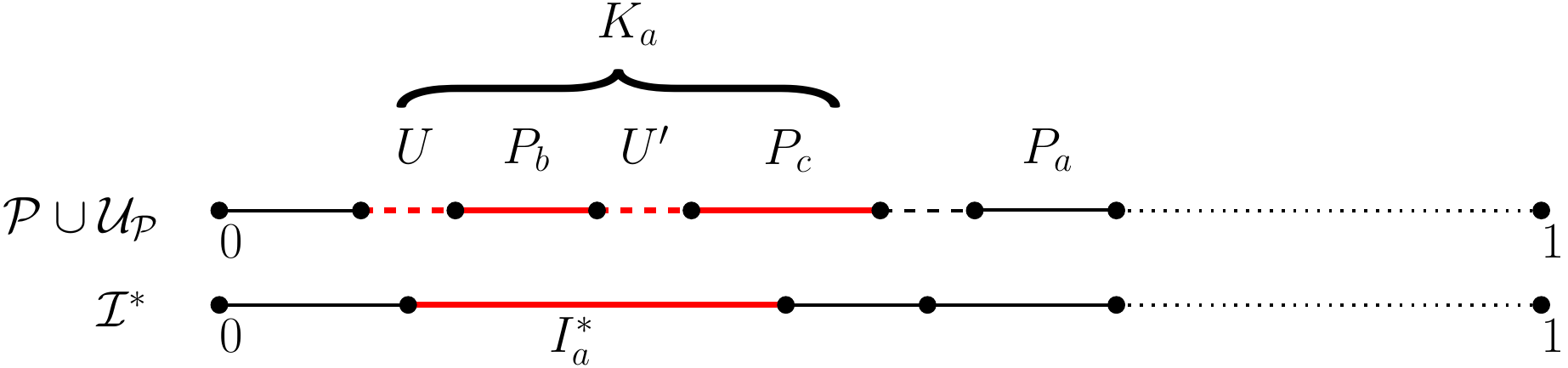}
\end{center}
\caption{In this example, $I^*_a$ intersects with four intervals from $\mathcal{P} \cup \mathcal{U}_{\mathcal{P}}$; in particular, $K_a = \{U, P_b, U', P_c \}$ with $k_a= 4$, and the right endpoints of three of them (namely, $U$, $P_b$, and $U'$) are contained in $I^*_a$, i.e., $\ell_a = 3$.}
\label{figure:overlap}
\end{figure}
Since $ \cup_{Q \in \mathcal{P} \cup \mathcal{U}_{\mathcal{P}}} Q = [0,1]$, interval $I^*_a$ is covered by the union of intervals in $K_a$; $I^*_a \subseteq \cup_{Q \in K_a} Q$. 

Lemma~\ref{lemma:partial-ef} gives us $v_a(P_a) \geq v_a(Q) - \varepsilon/n^2$ for all $Q \in K_a$; as before, $\varepsilon \in (0, 1/3]$ is the constant used in $\ALG$. In addition, recall that $v_a(I_a) \geq v_a(P_a)$ and $\frac{3\varepsilon}{n} v_a(I_a) \geq \varepsilon/n^2$ (see equation \eqref{ineq:half-prop-use}). These observations imply that  
\begin{align*}
\left( 1 + \frac{3\varepsilon}{n} \right) v_a(I_a) \geq v_a(Q) \qquad \text{for all } Q \in K_a
\end{align*}

Summing these inequalities and using the containment $I^*_a \subseteq \cup_{Q \in K_a} Q$ we get\footnote{Recall that the valuations are sigma additive.} 
\begin{align*}
\left(1 + \frac{3 \varepsilon}{n} \right) k_a \ v_a(I_a) & \geq v_a(I^*_a) \qquad \text{ for all } a\in [n].
\end{align*} 

Therefore, $\left(  \prod_{a=1}^n \left[ \left(1 + \frac{3\varepsilon}{n} \right) k_a \  v_a(I_a) \right]  \right)^{1/n}  \geq \left( \prod_a v_a(I^*_a) \right)^{1/n} = \NSW(\mathcal{I}^*)$. Simplifying this equation gives us   
\begin{align}
\NSW(\mathcal{I}^*) \leq \left(1 + \frac{3\varepsilon}{n} \right) \left( \prod_{a=1}^n  k_a \right)^{1/n} \ \left(  \prod_{a=1}^n  v_a(I_a) \right)^{1/n} = \left(1 + \frac{3\varepsilon}{n} \right) \left( \prod_{a=1}^n  k_a \right)^{1/n} \ \NSW(\mathcal{I}) \label{ineq:comp-nsw}
\end{align}
Next, we will show, via a simple counting argument, that $\sum_{a=1}^n k_a \leq 3n + 1$. For an interval $I^*_a$ in the Nash optimal allocation, write $\ell_a$ to denote the number of intervals in the collection $\mathcal{P} \cup \mathcal{U}_{\mathcal{P}}$ whose \emph{right} endpoint in contained in $I^*_a$. Note that $k_a \leq \ell_a + 1$ for all $a \in [n]$; at most one interval in $K_a$ ends after $I^*_a$. Therefore, $\sum_{a=1}^n k_a \leq n + \sum_{a=1}^n \ell_a \leq 3n+1$. The last inequality follows from the observation that the total number of right endpoints across intervals in $\mathcal{P} \cup \mathcal{U}_{\mathcal{P}}$ (i.e., $\sum_{a=1}^n \ell_a$) is at most $2n+1$, since each right endpoint is associated with exactly one interval from $\mathcal{P} \cup \mathcal{U}_{\mathcal{P}}$ (recall that $|\mathcal{P}| = n$ and $|\mathcal{U}_{\mathcal{P}}| \leq n+1$).   

The AM-GM inequality implies $\left( \prod_{a=1}^n k_a \right)^{1/n} \leq \frac{1}{n} \sum_{a=1}^n k_a \leq 3 + \frac{1}{n}$. This bound, along with inequality \eqref{ineq:comp-nsw},\footnote{Recall that $\varepsilon \leq 1/3$.} establishes the stated approximation guarantee: $$\left( 3 + \frac{5}{n} \right) \NSW ( \mathcal{I}) \geq \left( 3 + \frac{1}{n} \right) \left( 1 +  \frac{3\varepsilon}{n} \right) \NSW ( \mathcal{I}) \geq  \NSW(\mathcal{I}^*)$$. 
\end{proof}

We conclude this section by proving that the approximately envy-free allocations always have near-optimal Nash social welfare. Note that, directly invoking the following theorem for the $\left(3+o(1) \right)$-$\EF$ allocations computed by $\ALG$, one would essentially obtain an approximation ratio of six for the Nash social welfare maximization problem. Indeed, this guarantee is weaker than the one shown (via a tailored analysis) in Theorem~\ref{theorem:nsw}.   
 
\begin{restatable}{theorem}{TheoremEFNSW}
\label{theorem:ef-nsw}
In a cake-division instance, let $\widetilde{\mathcal{I}}$ be an $\alpha$-approximately envy-free allocation and $\mathcal{I}^*$ be a Nash optimal allocation. Then, 
\begin{itemize}
\item[(i)] $\widetilde{\mathcal{I}}$ provides a $2\alpha$-approximation to Nash social welfare, i.e., $\NSW(\widetilde{\mathcal{I}}) \geq \frac{1}{2\alpha} \NSW(\mathcal{I}^*)$.   
\item[(ii)] $\mathcal{I}^*$ is $4$-approximately envy-free. 
\end{itemize}
\end{restatable}
\begin{proof}
Let $I^*_a$ denote the interval assigned to agent $a \in [n]$ in the Nash optimal allocation $\mathcal{I}^* = \{I^*_1, \ldots, I^*_n\}$. Also, write $\widetilde{K}_a$ to denote the set of intervals in the $\alpha$-$\EF$ allocation $\widetilde{\mathcal{I}} = \{\widetilde{I}_1, \ldots, \widetilde{I}_n\}$ that intersect with $I^*_a$, i.e., $\widetilde{K}_a \coloneqq \{ \widetilde{I}_b \in \widetilde{\mathcal{I}} \mid \widetilde{I}_b \cap I^*_a \neq \emptyset \}$.\footnote{Here, we follow the above-mentioned convention that mandates two intervals to be disjoint, if they intersect exactly at an endpoint.} Let $\widetilde{k}_a$ denote the cardinality of this set, $\widetilde{k}_a \coloneqq |\widetilde{K}_a|$.   
Since $\widetilde{\mathcal{I}}$ is an allocation, we have $\cup_a \widetilde{I}_a = [0,1]$. Therefore, for all agents $a \in [n]$, interval $I^*_a$ is covered by the union of intervals in $\widetilde{K}_a$; $I^*_a \subseteq \cup_{\widetilde{I} \in \widetilde{K}_a} \widetilde{I}$. Also, the fact that $\widetilde{\mathcal{I}}$ is $\alpha$-$\EF$ implies $\alpha v_a( \widetilde{I}_a) \geq v_a( \widetilde{I}_b)$ for all $\widetilde{I}_b \in \widetilde{K}_a$. Summing these inequalities and using the containment $I^*_a \subseteq \cup_{\widetilde{I} \in \widetilde{K}_a} \widetilde{I}$ gives us\footnote{Recall that the valuations are sigma additive.} 
\begin{align*}
\alpha \widetilde{k}_a \ v_a( \widetilde{I}_a) \geq v_a(I^*_a) \qquad \text{ for all } a \in [n].
\end{align*} 

Therefore, $\left(  \prod_{a=1}^n \left( \alpha \widetilde{k}_a \  v_a( \widetilde{I}_a) \right)  \right)^{1/n}  \geq \left( \prod_a v_a(I^*_a) \right)^{1/n} = \NSW(\mathcal{I}^*)$. Simplifying we get   
\begin{align}
\NSW(\mathcal{I}^*) \leq \alpha \left( \prod_{a=1}^n  \widetilde{k}_a \right)^{1/n} \ \left(  \prod_{a=1}^n  v_a( \widetilde{I}_a) \right)^{1/n} = \alpha \left( \prod_{a=1}^n  \widetilde{k}_a \right)^{1/n} \ \NSW(\widetilde{\mathcal{I}}) \label{ineq:comp}
\end{align}
Next, we will show, via a simple counting argument, that $\sum_{a=1}^n \widetilde{k}_a \leq 2n$. For an interval $I^*_a$ in the Nash optimal allocation, write $\widetilde{\ell}_a$ to denote the number of intervals in allocation $\widetilde{\mathcal{I}}=\{\widetilde{I}_1, \ldots, \widetilde{I}_n\}$ whose \emph{right} endpoint in contained in $I^*_a$. Note that $\widetilde{k}_a \leq \widetilde{\ell}_a + 1$ for all $a \in [n]$; at most one interval in $\widetilde{K}_a$ ends after $I^*_a$. Therefore, $\sum_{a=1}^n \widetilde{k}_a \leq n + \sum_{a=1}^n \widetilde{\ell}_a = 2n$. The last inequality follows from the observation that the total number of right endpoints across intervals in $\widetilde{\mathcal{I}}$ (i.e., $\sum_{a=1}^n \widetilde{\ell}_a$) is exactly equal to $n$, since each right endpoint is associated with exactly one interval.   

The AM-GM inequality implies $\left( \prod_{a=1}^n \widetilde{k}_a \right)^{1/n} \leq \frac{1}{n} \sum_{a=1}^n \widetilde{k}_a \leq 2$. This bound, along with inequality \eqref{ineq:comp}, establishes the stated approximation guarantee: $2 \alpha \ \NSW ( \widetilde{\mathcal{I}}) \geq \NSW(\mathcal{I}^*)$. 

To prove the complementary part of the Theorem (part (ii)) assume, towards a contradiction, that $\mathcal{I}^*=\{I^*_1, \ldots, I^*_n\}$ is not $4$-$\EF$. That is, there exist agents $a$ and $b$ such that $v_a(I^*_a) < \frac{1}{4} v_a (I^*_b)$. The divisibility of valuations ensures that the interval $I^*_b$ can be  partitioned into two disjoint intervals $I'_b$ and $I''_b$ with the property that $v_b (I'_b) = v_b(I''_b) = \frac{1}{2} v_b(I^*_b)$.  

Furthermore, since the valuations are sigma additive, for agent $a$ either interval $I'_b$ or $I''_b$ is of value strictly greater than $2 v_a(I^*_a)$. Say, $v_a(I''_b) > 2 v_a (I^*_a)$. Now, consider a partial allocation $\mathcal{J}= \{J_1, \ldots, J_n\}$ obtained by setting $J_c = I^*_c$ for all $c \in [n] \setminus \{ a, b \}$, $J_a  = I''_b$ and $J_b = I'_b$.  Note that $ v_a(J_a) v_b(J_b) > v_a(I^*_a) v_b(I^*_b)$, hence we have $\NSW(\mathcal{J}) > \NSW(\mathcal{I}^*)$. Given that any partial allocation can be extended to an allocation without decreasing the Nash social welfare, the previous inequality contradicts the optimality of $\mathcal{I}^*$. This establishes part (ii) of the theorem and completes the proof. 
\end{proof}

\section{Approximation Algorithm for $\rho$-Mean Welfare Maximization}
\label{section:interval-scheduling}
This section addresses cake-division with the objective of maximizing the $\rho$-mean welfare. We obtain an approximation algorithm for this problem via a simple reduction to the weighted job interval selection problem (\JISP) \cite{erlebach2003interval, bar2001approximating}.  

A problem instance of $\JISP$ consists of a tuple $\langle [n], \{\mathcal{J}_i,w_i \}_{i \in [n] } \rangle$, where $n$ denotes the number of \emph{jobs},\footnote{In the developed reduction, the number of jobs will be set equal to the number of agents present in the cake-division instance, hence we overload $n$ to denote both of these quantities.} and for each job $i \in [n]$ we have $\mathcal{J}_i$, a collection of intervals in $[0,1]$. Here, every interval in $\mathcal{J}_i$ is endowed with a weight $w_i: \mathcal{J}_i \mapsto \mathbb{R}_+$. The goal of $\JISP$ is to select a collection of non-intersecting intervals such that at most one interval is selected from each $\mathcal{J}_i$ and the total weight of the collection is as large as possible.  Formally, a feasible solution to a $\JISP$ problem instance consists of a set of intervals $\mathcal{F} = \{ F_1, F_2,\ldots,F_n\}$ such that: (i)  all the selected intervals are pairwise disjoint, $F_i \cap F_j = \emptyset$ for all $i \neq j$. (ii) at most one interval is selected from each job: $F_i \in \mathcal{J}_i$ or $F_i$ is the empty interval, for all $i \in [n]$.

The objective of $\JISP$ is to find a feasible solution $\mathcal{F} =\{ F_1, \ldots, F_n\}$ that maximizes $\sum_{i \in [n]} w_i(F_i)$.\footnote{We follow the convention that the weight of the empty interval is equal to zero.} For a solution $\mathcal{F}$, we denote the value of this weight objective by $w(\mathcal{F})$.
         
The following lemma shows that, given a cake-division instance, we can construct a $\JISP$ instance such that the $\rho$-mean welfare is approximated by the weight objective. That is, the lemma presents an approximation-preserving reduction from cake division to interval scheduling. 
    \begin{restatable}{lemma}{LemmaDiscretization} \label{lemma:discretization}
   	Given a cake-division instance $\langle [n], \{v_a\}_a \rangle$ with piecewise-constant valuations along with parameters $\rho \in (0,1]$ and $\varepsilon \in (0,1]$, one can construct a $\JISP$ instance $\langle [n], \{\mathcal{J}_i,w_i \}_{i \in [n] } \rangle$,  in time $\left(\frac{n}{\varepsilon}\right)^{\mathcal{O} (1/ \rho)}$, such that  
   	\begin{enumerate}
   		\item If $\widehat{\mathcal{I}}=\{\widehat{I}_1, \ldots, \widehat{I}_n \}$ is an allocation that maximizes the $\rho$-mean welfare in the given cake-division instance, then there exists a feasible solution $\mathcal{F} = \{F_1, F_2, \ldots F_n\}$ of the $\JISP$ instance such that  		
   		\begin{align*}
   		\left( \sum_{i \in [n]} w_i (F_i) \right)^{1/\rho}  & \geq  \left(1 - \frac{\varepsilon}{n} \right)^{1 + \frac{1}{\rho}} \left( \sum_{a \in [n]} \left[v_a (\widehat{I}_a)\right]^{\rho} \right)^{1/\rho} 
   		\end{align*} 		
   		\item For every feasible $\JISP$ solution $\mathcal{F} = \{F_1, F_2, \ldots F_n\}$, there exists a partial allocation $\I=\{I_1, \ldots, I_n\}$ in the cake-division instance with the property that $ \left( \sum_{a} \left[ v_a (I_a) \right]^{\rho} \right)^{1/\rho}  = \left( \sum_{i} w (F_i) \right)^{1/\rho}  $. Furthermore, partial allocation $\I$ can be computed from the given JISP solution $\mathcal{F}$ in polynomial time. 
   	\end{enumerate}
   \end{restatable}
\begin{proof}
First, we describe the construction of a $\JISP$ instance from the given cake-division instance. The reduction is based on discretizing the cake. 
   
Let $\delta := \left(\frac{\varepsilon}{n^2} \right)^{1/\rho}$, where $\varepsilon \in (0,1]$ is a given parameter. We find a set of $\mathcal{O}\left( \left( \varepsilon^{-1} \ {n^2} \right)^{1 + \frac{1}{\rho}} \right)$ points (cake cuts) $0= x_0 < x_1 < x_2 < \ldots, x_k < x_{k+1}=1$ such that for every index $\ell \in \{0,1, \ldots, k \}$ and each agent $a \in [n]$ we have $v_a([x_\ell, x_{\ell+1}]) \le \frac{\delta \varepsilon}{2n}$. The procedure to find these points starts by initializing $x_0 = 0$. Then, iteratively, for each $\ell \geq 0$, we find $x_{\ell + 1}$ using $x_\ell$: for each agent $a$, find the smallest value $c_a \in (x_\ell, 1]$ such that $v_a([x_\ell, c_a]) = \frac{\delta \varepsilon}{2n}$.\footnote{This step can be performed efficiently in the Robertson-Webb model, as well as for piecewise-constant valuations.} If for an agent $a$ we have $v_a([x_\ell, 1]) < \frac{\delta \varepsilon}{2n}$, then assign $c_a =1$. We set $x_{\ell+  1} \coloneqq \min_{a \in [n]} c_a$ and continue as long as $x_{\ell + 1} <1$. Write $\mathcal{X}\coloneqq \{x_0, x_1, \ldots, x_{k+1} \} $ to denote this set of points and note that the normalization of the agents' valuations ensures that the cardinality of $\mathcal{O}\left( \left( \varepsilon^{-1} \ {n^2} \right)^{1 + \frac{1}{\rho}} \right)$.

The reduction to $\JISP$ is as follows. For each agent $a \in [n]$, we associate a job with the same index in the $\JISP$ instance $\langle [n], \{\mathcal{J}_a,w_a \}_{a \in [n] } \rangle$. Furthermore, for each $a \in [n]$, the set $\mathcal{J}_a$ is defined to be collection of all intervals with endpoints in the computed set of cuts $\mathcal{X}$, i.e., $\mathcal{J}_a \coloneqq \{[x_\ell,x_r]  \mid  0 \le \ell < r \le k+1 \text{ and } \ell,r \in \mathbb{Z} \}$. Note that the set of intervals is the same for all jobs. Each weight function $w_a$ is defined as $ w_a([x_\ell,x_r])  \coloneqq \left( v_a([x_\ell,x_r]) \right)^\rho$  for all intervals $[x_\ell, x_r] \in \mathcal{J}_a$.
      
Now, we will prove the first part of the Lemma. Consider $\widehat{I} = \{\widehat{I}_1,\widehat{I}_2,\cdots,\widehat{I}_n \}$, an allocation that maximizes the $\rho$-mean welfare. From $\widehat{I}$, we obtain a feasible solution $\mathcal{F} = \{F_1, \ldots, F_n \}$ by setting $F_a$, for each $a \in [n]$, to be the largest interval of the form $[x_\ell,x_r] \in \mathcal{J}_a$ that is contained within $\widehat{I}_a$. Specifically, if $\widehat{I}_a = [y,z]$, then $F_a = [x_\ell,x_r] \in \mathcal{J}_a$ where $\ell \coloneqq \min_{x_i \geq y} i$ and $r \coloneqq \max_{x_j \leq z} j$.\footnote{If $\widehat{I}_a$ is the empty interval, then so is $F_a$.}  Excluding the (trivial) corner cases, we have $x_{\ell - 1} < y \leq x_\ell$ and $x_{r} \leq z < x_{r+1}$. Note that $F_a$ is obtained by removing subintervals at the two ends of $\widehat{I}_a$. These removed subintervals are of value at most $v_a([x_{\ell-1}, x_{\ell}]) \leq \frac{\delta \varepsilon}{2n}$ and $v_a([x_{r}, x_{r+1}]) \leq \frac{\delta \varepsilon}{2n}$, respectively.  Hence, using the fact that the valuations are sigma additive, we obtain the following bound on the value lost in discretizing $\widehat{I}_a$ to obtain $F_a$:
       \begin{align}
        \label{ineq:valueloss}
      v_a(\widehat{I}_a) - v_a(F_a) \le 2 \times \frac{\delta \varepsilon}{2n} = \frac{\delta \varepsilon}{n}
    \end{align}
Since the intervals in the optimal allocation $\widehat{I}$ are pairwise disjoint and each $F_a \in \mathcal{J}_a$ (or $F_a$ is the empty interval), the collection $\mathcal{F} = \{F_1, F_2, \cdots F_n\}$ is a feasible solution of the $\JISP$ instance. 
 
For analysis, we partition the set of agents (and correspondingly the jobs) into low-valued and high-valued agents. In particular, define $\mathcal{S}_1 \coloneqq \{a \in [n] \mid v_a(\widehat{I}_a) \le \delta\}$ and $\mathcal{S}_2 =[n] \setminus \mathcal{S}_1$ to be the remaining agents. 

For all agents $a \in \mathcal{S}_1$, by definition, we have 
    \begin{align}
    \label{ineq:small}
        \left[v_a(\widehat{I}_a)\right]^\rho \le \left(\Big(\frac{\varepsilon}{n^2}\Big)^{1/\rho}\right)^\rho = \frac{\varepsilon}{n^2}
    \end{align}
    \noindent
    For all jobs $a \in \mathcal{S}_2$ 
   \begin{align}
        w_a(F_a) &= \left[v_a(F_a)\right]^\rho \nonumber \\
                &\geq \left[v_a(\widehat{I}_a) - \frac{\delta \varepsilon}{n} \right]^\rho  \tag{via \eqref{ineq:valueloss}} \nonumber\\
                &\geq \left[ v_a(\widehat{I}_a) \left( 1 -\frac{\varepsilon}{n} \right) \right]^\rho \tag{since $v_a(\widehat{I}_a) \ge \delta$} \nonumber\\     
                &\geq \left(1-\frac{\varepsilon}{n}\right)^\rho \left[v_a (\widehat{I}_a)\right]^\rho    \label{ineq:large}
    \end{align} 
    Let $T$ be the optimal $\rho$-mean welfare in the given cake-division instance, i.e., $T= \mathrm{M}_\rho(\widehat{I}) = \left( \frac{1}{n} \sum_{a=1}^n \left[v_a(\widehat{I}_a)\right]^\rho \right)^{1/\rho}$. Note that $T \ge 1/n$, since the $\rho$-mean welfare of a proportional division of the cake is at least $1/n$. Given that $\mathcal{S}_1$ and $\mathcal{S}_2$ partition $[n]$, we have  
       $ \frac{1}{n} \sum_{a \in \mathcal{S}_1} \left[v_a(\widehat{I}_a)\right]^\rho +  \frac{1}{n} \sum_{a \in \mathcal{S}_2} \left[v_a(\widehat{I}_a)\right]^\rho    = T^\rho$.
Using inequality \eqref{ineq:small}, we obtain a lower bound on the contribution of the agents in $\mathcal{S}_2$ to the $\rho$-mean welfare
    \begin{align*}
        \frac{1}{n} \sum_{a \in \mathcal{S}_2} \left[v_a(\widehat{I}_a)\right]^\rho  &\ge T^\rho - \frac{1}{n} \  n \ \frac{\varepsilon}{n^2} 
         \ge \left(1- \frac{\varepsilon}{n}\right) T^\rho \tag{since $T \ge 1/n$ and $\rho \in (0,1]$}
    \end{align*}
    Equation \eqref{ineq:large} connects this lower bound to the weight of the feasible solution $\mathcal{F}$ 
        \begin{align*}
        \frac{1}{n} \sum_{a \in \mathcal{S}_2} w_a(F_a) &\ge \left(1 -\frac{\varepsilon}{n}\right)  \left(1- \frac{\varepsilon}{n}\right)^\rho T^\rho 
        =  \left(1 -\frac{\varepsilon}{n}\right)  \left(1- \frac{\varepsilon}{n}\right)^\rho \frac{1}{n} \left(\sum_{a=1}^n \left[v_a(\widehat{I}_a)\right]^\rho \right) 
        \end{align*}
 Multiplying both sides of this inequality by $n$, we get\footnote{Recall that the weights of all the intervals in the $\JISP$ instance are nonnegative.}
    \begin{align*}
        \left(\sum_{a \in [n]} w_a(F_a)\right) &\ge \left(1 -\frac{\varepsilon}{n}\right)  \left(1- \frac{\varepsilon}{n}\right)^\rho \left(\sum_{a=1}^n \left[v_a(\widehat{I}_a)\right]^\rho \right)
    \end{align*}    
    Finally, we obtain the desired inequality by exponentiating both sides of the previous inequality to the power $\frac{1}{\rho}$
        \begin{align*}
        \left(\sum_{a \in [n]} w_a(F_a)\right)^{1/\rho} &\ge \left(1 -\frac{\varepsilon}{n}\right)  \left(1- \frac{\varepsilon}{n}\right)^\frac{1}{\rho} \left(\sum_{a=1}^n \left[v_a(\widehat{I}_a)\right]^\rho \right)^\frac{1}{\rho} 
    \end{align*}
This concludes the proof of the first part of the Lemma.
    
    For the second part of the Lemma, note that every feasible solution $\mathcal{F} = \{F_1,F_2, \cdots , F_n\}$ of the $\JISP$ instance can also be mapped directly to an (partial) allocation $\mathcal{I} =\{I_1, \ldots, I_n \}$ of the cake: each agent $a$ receives the interval $I_a = F_a \subseteq [0,1]$, which (by feasibility of $\mathcal{F}$) does not intersect with any other agent's interval. Also, by definition of the weights, $w_a(F_a) = \left[v_a(F_a)\right]^\rho = \left[v_a(I_a)\right]^\rho$, which gives us the desired equality.
    
\end{proof}

   This lemma allows us to directly invoke the result of Bar-Noy et al. \cite{bar2001approximating} (which provides a polynomial-time $2$-approximation algorithm for $\JISP$) to obtain an approximation algorithm for the $\rho$-mean maximization problem. 
   The main result of this section is stated in the following theorem. 
   
   \begin{restatable}{theorem}{IntervalApprox}
   \label{theorem:interval-approx}
   	For $\rho \in (0,1]$, $\varepsilon \in (0,1)$ and cake-division instances $\C = \left\langle [n], \{v_a\}_{a \in [n]} \right\rangle$ with piecewise-constant valuations, there exists an algorithm that---in time $\left(\frac{n}{\varepsilon}\right)^{\mathcal{O} (1/\rho)}$---finds a  $\left(2 + \frac{4 \varepsilon e}{n} \right)^{\frac{1}{\rho}}$-approximation to the $\rho$-mean welfare maximization problem.
	\end{restatable}
\begin{proof}
For the given cake-division instance, we instantiate Lemma~\ref{lemma:discretization} to construct a $\JISP$ instance $\langle [n], \{\mathcal{J}_i,w_i \}_{i \in [n] } \rangle$ in time $\left(\frac{n}{\varepsilon}\right)^{\mathcal{O} (1/\rho)}$. 
    
    Let $\widehat{\mathcal{I}} =\{\widehat{I}_1,\ldots, \widehat{I}_n\}$ be an allocation in the cake-division instance that maximizes the $\rho$-mean welfare. The first part of Lemma~\ref{lemma:discretization} asserts that there there exists a feasible solution $\mathcal{F}$ in the constructed $\JISP$ instance such that 
    \begin{align}
        \left(w(\mathcal{F})\right)^\frac{1}{\rho}  \ge  \left(1-\frac{\varepsilon}{n}\right)^{ \frac{1}{\rho} + 1} \left( \sum_{a \in [n]} \left[v_a (\widehat{I}_a)\right]^{\rho} \right)^{1/\rho} & \geq  e^{-\frac{\varepsilon}{n} \left( \frac{1}{\rho} + 1 \right) } \left( \sum_{a \in [n]} \left[v_a (\widehat{I}_a)\right]^{\rho} \right)^{1/\rho} \nonumber \\
             & \geq  e^{-\frac{\varepsilon}{n} \frac{2}{\rho}} \left( \sum_{a \in [n]} \left[v_a (\widehat{I}_a)\right]^{\rho} \right)^{1/\rho} 
        \label{ineq:part1}
    \end{align}
    The last inequality follows from the fact that $1/\rho \geq 1$. 
    
    Let $\mathcal{F}^*$ be an optimal solution of the $\JISP$ instance, i.e., $\mathcal{F}^*$ has the maximum possible weight, $w(\mathcal{F}^*)$, among all feasible solutions.  
     
    The  algorithm of Bar-Noy et al.~\cite{bar2001approximating} achieves an approximation ratio of $2$ for $\JISP$, i.e., it efficiently computes a feasible schedule $\mathcal{G} = \{G_1,G_2,\ldots,G_n\}$ which satisfies, $w(\mathcal{G}) \geq \frac{1}{2} w(\mathcal{F}^*)$. Note that, using the second part of Lemma~\ref{lemma:discretization}, we can efficiently find an allocation $\mathcal{I}=\{I_1, \ldots,  I_n\}$ (in the underlying cake-division instance) such that $\left( \sum_{a \in [n]} \left[v_a (I_a)\right]^{\rho} \right) = w(\mathcal{G})  \geq \frac{1}{2} w(\mathcal{F}^*)$. Exponentiating both sides of the previous inequality to the power $\frac{1}{\rho}$ gives us 
    
    \begin{align*}
        \left( \sum_{a \in [n]} \left[v_a (I_a)\right]^{\rho} \right)^\frac{1}{\rho} \ge \left(\frac{1}{2}\right)^\frac{1}{\rho} \left(w(\mathcal{F}^*)\right)^\frac{1}{\rho} 
      \ge  2^{-\frac{1}{\rho}}  e^{-\frac{2 \varepsilon}{n \rho}} \ \left( \sum_{a \in [n]} \left[v_a (\widehat{I}_a)\right]^{\rho} \right)^{1/\rho} 
        \end{align*}
 The last inequality follows from \eqref{ineq:part1} and the optimality of $\mathcal{F}^*$. To obtain the desired approximation guarantee we divide both sides of the previous equation by $n^\frac{1}{\rho}$       
        \begin{align*}
        \mathrm{M}_\rho(\mathcal{I}) &\ge  2^{-\frac{1}{\rho}} \ e^{-\frac{2 \varepsilon}{n \rho}} \  \mathrm{M}_\rho(\widehat{\mathcal{I}}) \\ 
        & = \left( 2 \ e^{2 \varepsilon/n} \right)^{-1/\rho}  \mathrm{M}_\rho(\widehat{\mathcal{I}}) \\
        & \geq \left( 2 + \frac{4 \varepsilon e}{n} \right)^{-1/\rho} \left(1-\frac{\varepsilon}{n} \right) \mathrm{M}_\rho(\widehat{\mathcal{I}}) \tag{since $1/\rho \geq 1$ and $2 + \frac{4 \varepsilon e}{n} \geq 2 e^{2 \varepsilon/n}$ for $n \geq 2$}
    \end{align*}
Therefore, the computed allocation $\mathcal{I}$ achieves the stated approximation ratio of $\left(2 + \frac{4 \varepsilon e}{n} \right)^{1/\rho}$. 

For establishing the time complexity of this algorithm, note that the size of the constructed $\JISP$ instance is $\left(\frac{n}{\varepsilon}\right)^{\mathcal{O} (1/\rho)}$. Hence, the algorithm of Bar-Noy et al.~\cite{bar2001approximating} would run in time that is polynomial in $\left(\frac{n}{\varepsilon}\right)^{\mathcal{O} (1/\rho)}$, i.e., the specified running-time bound follows. 
\end{proof}

Note that for any constant $\rho \in (0,1]$, Theorem~\ref{theorem:interval-approx} provides a constant-factor approximation algorithm that runs in polynomial time. In particular, for the $\rho=1$ case (i.e., for average social welfare), we obtain a polynomial-time $\left(2 + o(1) \right)$-approximation algorithm. As mentioned previously, this instantiation improves upon the $8$-approximation guarantee obtained specifically for social welfare in the work of Aumann et al.~\cite{aumann2013computing}.

\section{Hardness of Maximizing Nash Social Welfare}
\label{section:hardness}

This section establishes the $\APX$-hardness of finding cake divisions (with connected pieces) that maximize Nash social welfare. That is, we show that, for a fixed constant $ c \in (0,1)$, it is $\NP$-hard to find an allocation (i.e., a cake division with connected pieces) whose Nash social welfare is within $c$ times the optimal. 
Appendix \ref{section:rhomeanhardness} details an analogous hardness result for the $\rho$-mean welfare objective.

 \begin{theorem} \label{theorem:nswhardness}
Given a cake-division instance with piecewise-constant valuations, the problem of computing an allocation that maximizes Nash social welfare is $\APX$-hard. 
\end{theorem}	 

We prove this theorem by developing a (gap-preserving) reduction from the $\textrm{Gap 3-SAT-5}$ problem~\cite{arora1999polynomial, arora1998proof}. Section~\ref{section:gadget} presents the key gadget used in the reduction and the proof of Theorem \ref{theorem:nswhardness} appears in Section~\ref{section:hardnessproof}. 

The $\textrm{Gap 3-SAT-5}$ problem, with parameter $\alpha \in (0,1)$, is defined as follows\\
\noindent
\textit{Input:} A Boolean formula $\phi$ in conjunctive normal form; in particular, $\phi$ is specified as a conjunction of a set of clauses, $\textrm{C} = \{ C_1, C_2, \ldots, C_m \}$, defined over Boolean variables $\textrm{X}=\{x_1, \ldots, x_r \}$. Here, each clause $C_j$ is a disjunction of at most three literals and every variable $x_i$ appears in at most five clauses in $\phi$ (either as literal $x_i$ or as its negation $\overline{x}_i$).\\
\noindent
\textit{Objective:} Distinguish between the following two cases \\
{YES:} $\phi$ is satisfiable. \\
{NO:} No assignment of the variables satisfies more than $(1-\alpha)$ fraction of the clauses in $\phi$. \\

$\textrm{Gap 3-SAT-5}$ is a promise problem: the given Boolean formula $\phi$ is guaranteed to satisfy either the YES case or the NO case. It is shown in~\cite{arora1999polynomial, arora1998proof} that there exists a constant $\alpha \in (0,1)$ for which $\textrm{Gap 3-SAT-5}$ is $\NP$-hard, i.e., for a specific constant $\alpha \in (0,1)$, it is $\NP$-hard to distinguish whether a given instance of $\textrm{Gap 3-SAT-5}$ satisfies the YES case or the NO case. 

We develop a gap-preserving reduction from $\textrm{Gap 3-SAT-5}$ to the Nash social welfare maximization problem. That is, in our reduction, if the given $\textrm{Gap 3-SAT-5}$ instance satisfies the YES case, then the Nash social welfare in the constructed cake-division instance will be above a threshold, say $\tau$. Complementarily, in the NO case, the Nash social welfare will be below $c (\alpha) \ \tau$, for a fixed constant $c(\alpha) \in (0,1)$, which depends only on the underlying gap parameter $\alpha$. Therefore, using a $\left(1/c(\alpha)\right)$-approximation algorithm for the Nash social welfare maximization problem, one can distinguish between the YES and the NO cases of $\textrm{Gap 3-SAT-5}$. Since the latter problem is $\NP$-hard, such an approximation algorithm does not exist, unless ${\rm P} = {\rm NP}$. That is, we obtain $\left(1/c(\alpha)\right)$-inapproximability of maximizing Nash social welfare and, hence, the stated $\APX$-hardness result holds. 

\subsection{Construction of a Cake-Division Instance} \label{section:gadget}
In this section, starting with an instance $\phi$ of $\textrm{Gap 3-SAT-5}$, we will construct a cake-division instance $\C(\phi)$ which achieves the above-mentioned gap property for Nash social welfare. This construction is used in Section \ref{section:hardnessproof} to complete the reduction and prove Theorem \ref{theorem:nswhardness}. 

Let $r$ and $m$, respectively, denote the number of variables and the number of clauses in a given Boolean formula $\phi$. Recall that for $\textrm{Gap 3-SAT-5}$ instances each variable occurs in at most five clauses. Furthermore, since each clause in $\phi$ contains at most three literals, we have $r \leq 3m$.

In the reduction, for ease of presentation, the constructed cake will correspond to the segment $[0, 14r+1]$ (it can be rescaled to $[0,1]$ without affecting the arguments). The cake will be composed of $r$ pairwise-disjoint intervals,  $H_1,\ldots, H_r$, each of length $14$, along with an auxiliary interval $G$ (placed at the end of the cake) of length one. That is, the cake is obtained by concatenating $H_i$s and $G$. For all $i \in [r]$, we associate variable $x_i$  with the interval $H_i \coloneqq [14(i-1), 14i]$. The cake-division instance will have one agent for every clause $C_j$, with $j \in [m]$. In addition, the interval $H_i$ is itself partitioned into $14$ unit-length subintervals, $\{e^i_k\}_{k=1}^{14}$ (see Figure~\ref{figure:gadget}). A notable property of the reduction is that each subinterval $e^i_k$ gets valued by at most one agent. 

The cake-division instance $\mathcal{C}(\phi)$ will have four sets of agents: (1) separator agents, (2) base agents, (3) clause agents and, (4) an auxiliary agent.  We will say that an agent $a$'s value within an interval $I$ is \emph{constant} iff $a$'s density function is constant throughout $I=[\ell,r]$, i.e., we have $v_a(I) = w$ and this value is said to be constant within $I$ iff the density function, $\nu_a$, satisfies $\nu_a(x) = \frac{w}{\left(r-\ell\right)}$ for all $x \in I=[\ell,r]$.  Next, we specify the valuations of all the agents based on their type. 
\begin{itemize}[leftmargin=1.7mm]
\item \textbf{Separator Agents:} We include $2r$ separator agents--two separator agents $s_i$ and $s'_i$ for each interval $H_i$. The density function of both $s_i$ and $s'_i$ is nonzero only within $H_i$. Specifically, $v_{s_i}(e^i_7) =1$ and the value of $s_i$ is constant within this subinterval $e^i_7$. Also, we set $v_{s'_i}(e^i_{14}) =1$ and this value is constant within $e^i_{14}$.
	
The separator agent $s_i$ is introduced in the reduction to ensure that, in any Nash optimal allocation, no single agent receives both the left and right halves of $H_i$.  At a high level, the left-half of $H_i$ represents the positive occurrences of the variable $x_i$ in $\phi$ and right-half represents the negated occurrences. Similarly, the presence of the separator agent $s'_i$ ensures that, in any Nash optimal allocation, no agent receives an interval which nontrivially intersects with both $H_i$ and $H_{i+1}$.\footnote{The separator agent $s'_r$ accomplishes this goal for interval $H_r$ and the auxiliary interval $G$.}
   \begin{figure}[h]
   	\begin{center}
   		\includegraphics[scale=.75]{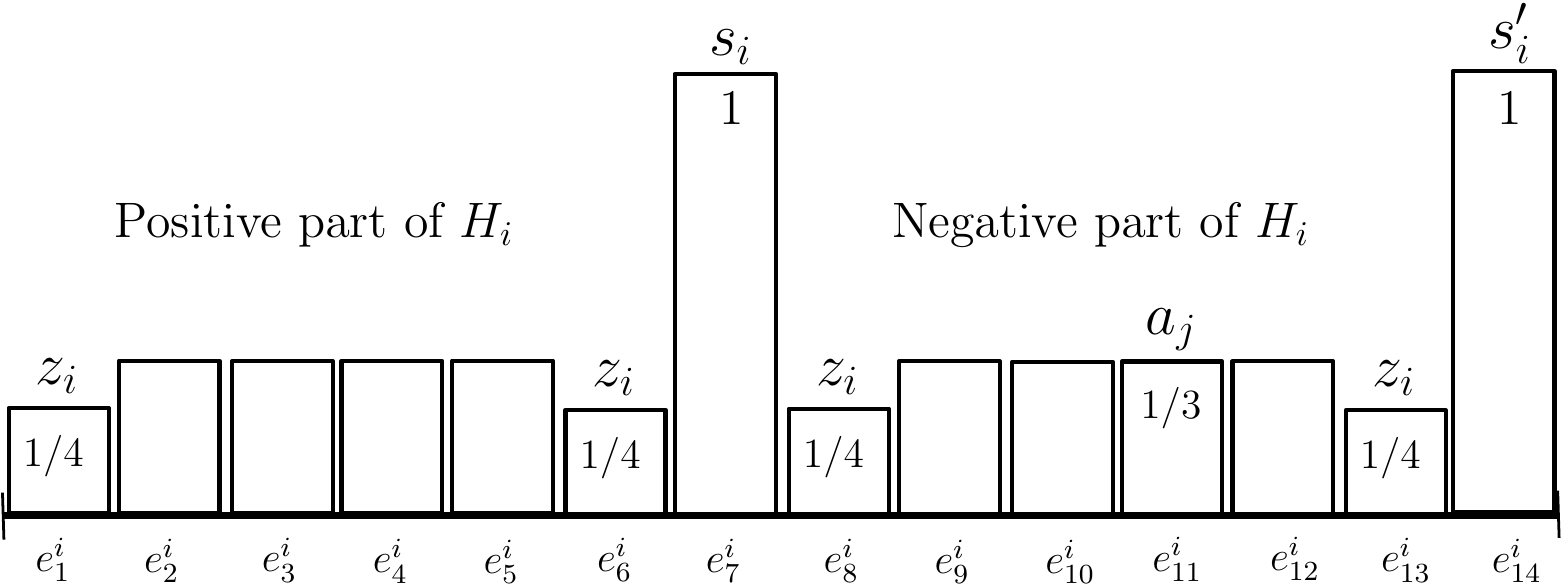}
   	\end{center}
   \caption{The figure depicts interval $H_i$ associated with variable $x_i$. The height of a vertical bar over each subinterval  $e^i_k$ depicts the total value of $e^i_k$ for the agent mentioned above the bar. As an illustration, the value of a clause agent, $a_j$, is also depicted; here, clause $C_j$ is the third negative occurrence of $x_i$ (in particular, $C_j$ contains the literal $\overline{x}_i$).}
   \label{figure:gadget}
\end{figure}
    
\item \textbf{Base Agents:} We include $r$ base agents--one base agent $z_i$ for each interval $H_i$, with $i \in [r]$. The total value of $z_i$ is spread over four subintervals in $H_i$. Specifically, $v_{z_i} ( S) = 1/4$ for each subinterval $S \in \{e^i_1, e^i_6, e^i_8, e^i_{13}\}$ and the value of $z_i$ is constant within each of these four subintervals.  
    
\item \textbf{Clause Agents:} We consider $m$ clause agents--one clause agent $a_j$ for each clause $C_j$, with $j \in [m]$. Here, agent $a_j$ values a subinterval in $H_i$ iff clause $C_j$ contains the variable $x_i$. Hence, the number of subintervals valued by $a_j$ is equal to the number of literals in clause $C_j$. For each variable $x_i$, we (arbitrarily) index the clauses (in $\phi$) that contain it as a positive literal, i.e., define $P_i \coloneqq \{C^i_1, C^i_2, \ldots, C^i_{p^*} \mid \text{literal } x_i \text{ appears in } C^i_k \text{ for all } k \in [p^*] \}$; here $p^*$ is the total number of clauses that contains $x_i$ as a positive literal. 
Analogously, index the clauses (in $\phi$) that contain the negated literal $\overline{x_i}$, $N_i \coloneqq \{C^i_1, C^i_2, \ldots, C^i_{n^*} \mid \text{literal } \overline{x}_i \text{ appears in } C^i_k \text{ for all } k \in [n^*] \}$. Given that at most five clauses contain $x_i$, we have $p^* + n^* \leq 5$. Also, since both literals $x_i$ and $\overline{x}_i$ must occur in $\phi$, we have $ p^* \geq 1 $ and $n^* \geq 1$. Equivalently, $p^* \leq 4$ and $n^* \leq 4$; this bound enables us to consider at most four subintervals $\{e^i_2, \ldots, e^i_5\}$ and $\{e^i_8, \ldots, e^i_{12}\}$ in the two halves of $H_i$, respectively. With this indexing in hand, we will refer to a clause $C_j$ as the $q$th positive occurrence of $x_i$ iff $C_j = C^i_q \in P_i$. Similarly, we will refer to a clause $C_k$ as the $q$th negative occurrence of $x_i$ iff  $C_k = C^i_q \in N_i$. 
	
As mentioned above, to set the valuation of the clause agent $a_j$ we consider the variables that occur in the clause $C_j$. For each variable $x_i$ in $C_j$, if $C_j$ is the $q$th positive occurrence of $x_{i}$ then we set $v_{a_j} (e^{i}_{1+q}) = 1/3$. Otherwise, if $C_j$ is the $q$th negative occurrence of variable $x_{i}$ then we set $v_{a_j} (e^{i}_{8+q}) = 1/3$. If $C_j$ contains less than three variables, then $a_j$ is made to positively value the auxiliary interval $G$ by an amount that ensures normalization $v_{a_j}([0, 14r+1]) = 1$. In particular, in case $C_j$ is a disjunction of $1 \leq t \leq 3$ literals, we set $v_{a_j} (G) = 1 - t/3$.

In light of this construction, we will henceforth refer to the left half of the interval $H_i$, $\{e^i_1, \ldots, e^i_6 \}$, to be the \emph{positive part}, while the right half, $\{e^i_8, \ldots, e^i_{13}\}$, to be the \textit{negative part} of $H_i$.	
	\item \textbf{Auxiliary Agent:} Finally, we have one auxiliary agent $d$ whose total valuation is confined to auxiliary interval $G$, i.e., $v_d(G) = 1$ and this value is constant within $G$.
\end{itemize}
In summary, we have constructed in polynomial-time a cake instance $\C(\phi)$ with $2r$ separator agents, $r$ base agents, $m$ clause agents and $1$ auxiliary agent; hence $n = 3r+m+1$. Also, note that in the constructed instance, the valuation of every agent $a$ is piecewise-constant and normalized, $v_a([0, 14r+1]) = 1$. 

\subsection{Proof of Theorem \ref{theorem:nswhardness}} \label{section:hardnessproof}
Consider an instance of the $\textrm{Gap 3-SAT-5}$ problem, with Boolean formula $\phi$ and parameter $\alpha$. Recall that, $\textrm{Gap 3-SAT-5}$ is a promise problem wherein it is $\NP$-hard to distinguish between the {YES case} ($\phi$ is satisfiable) and the {NO case} (no assignment satisfies more than $(1-\alpha)$ fraction of the clauses in $\phi$). Denote the set of variables in $\phi$ by $\textrm{X} = \{x_1, x_2, \dots, x_r\}$ and the set of clauses by $\textrm{C} = \{ C_1, C_2, \dots, C_m \}$. Construct a cake-division instance $\mathcal{C}(\phi)$ as detailed in Section \ref{section:gadget}. Independent of the underlying case (YES or NO), the following claim ensures the existence of a Nash optimal allocation for $\mathcal{C}({\phi})$ with useful structural properties.

\begin{restatable}{claim}{ClaimStructuralProperties}
\label{claim:structure}
The cake-division instance $\mathcal{C}(\phi)$ admits a Nash optimal allocation $\mathcal{I}=\{I_a\}_a$ with the property that:\footnote{Recall that $I_a$ denotes the interval assigned to agent $a$, under allocation $\mathcal{I}=\{I_a\}_a$.} 
\begin{itemize}
\item[(i)] For each $i \in [r]$, base agent $z_i$ is assigned exactly one of the following intervals: $e^i_1$, $\bigcup_{j=1}^6 e^i_j$, $e^i_8$, $\bigcup_{j=8}^{13} e^i_j$
\item[(ii)] For each $i \in [r]$, $I_{s_i} = e^i_7$ and $I_{s'_i} = e^i_{14}$ 
\item[(iii)] $\NSW(\mathcal{I}) > 0$
\end{itemize}
\end{restatable}
The proof of Claim \ref{claim:structure} is deferred to Appendix~\ref{section:proof-of-claim-1}. We will now proceed to prove Theorem \ref{theorem:nswhardness}. 

First, suppose that $\phi$ satisfies the YES case, i.e., there exists a Boolean assignment $f: \textrm{X} \to \{0,1\}$ to the variables that satisfies all clauses in $\phi$. Using $f$, we define a partial allocation $\mathcal{A}=\{A_a\}_a$ that achieves ``high'' Nash social welfare. In particular, for all $i \in [r]$, assign agent $s_i$ the subinterval $e^i_7$ (i.e., $A_{s_i} = e^i_7$ ) and assign agent $s'_i$ the subinterval $e^i_{14}$. In addition, allocate agent $z_i$ the interval $\bigcup_{j=1}^6 e^i_j$, if $f(x_i) = 0$. Otherwise, if $f(x_i) = 1$, assign the interval $\bigcup_{j=8}^{13} e^i_j$ to $z_i$. That is, $z_i$ is allocated the positive part of $H_i$ if $f(x_i) = 0$ and the negative part of $H_i$, if $f(x_i) = 1$. Since $f$ satisfies all the clauses, every clause $C_j$ either contains a literal $x_i$ for which $f(x_i) = 1$ or a literal $\overline{x}_k$ with $f(x_k) = 0$. Therefore, one of the following alternatives holds, for an index $q \in \{1,2,3,4\}$: $C_j$ is the $q$th positive occurrence of variable $x_i$ with $f(x_i) = 1$ or $C_j$ is the $q$th negative occurrence of variable $x_k$ with $f(x_k) = 0$. If the former condition holds, then allocate the interval $e^i_{1+q}$ to the clause agent $a_j$. Otherwise, if the latter condition holds, assign the interval $e^k_{8+q}$ to agent $a_j$. Finally, allocate interval $G$ to auxiliary agent $d$, i.e., set $A_d = G$. One can directly verify that $\mathcal{A}$ is a well-defined partial allocation in $\mathcal{C}({\phi})$. Also, since the valuations of the agents are nonnegative, $\mathcal{A}$ can be extended to a full allocation without reducing the value achieved by any agent. Let $NSW_Y$ denote the optimal Nash social welfare value obtained in $\mathcal{C}({\phi})$, under the current assumption that $\phi$ satisfies the YES case. Noting the values obtained by all the agents under $\mathcal{A}$, we get
\begin{align}
\label{equation:YEScase}
    NSW_Y \geq \left(1^{2r} \frac{1}{2^r} \frac{1}{3^m} 1\right)^{\frac{1}{3r+m+1}}
\end{align}

To complete the proof we now consider the complementary setting wherein $\phi$ falls under the NO case, i.e., under any assignment, at least $\alpha m$ clauses remain unsatisfied. In $\mathcal{C}(\phi)$, let $\mathcal{I} =\{I_a\}_a$ be a Nash optimal allocation that satisfies Claim~\ref{claim:structure}. With the NO case in hand, write $NSW_N$ to denote the optimal value of the Nash social welfare in $\mathcal{C}({\phi})$, i.e., $NSW_N = \NSW(\mathcal{I})$. By property (ii) of Claim~\ref{claim:structure}, we have that $I_{s_i} = e^i_7$ and $I_{s'_i} = e^i_{14}$, for all $i \in [r]$. This allocation to the separator agents ensures that the interval assigned to any other agent $a$, $I_a$, is contained in exactly one of the following intervals: the positive part of $H_i$, or the negative part of $H_i$, or the auxiliary interval $G$. 

Define the following  index sets for the base agents and the clause agents, respectively, $T \coloneqq \{i \in [r] \mid I_{z_i} = e^i_1 \text{ or } I_{z_i} = e^i_8\}$ and $P \coloneqq \{j \in [m] \mid I_{a_j} \subset G\}$. Write $t \coloneqq |T|$ and $p \coloneqq |P|$. By property (i) in Claim~\ref{claim:structure}, we can conclude that the base agents $z_i$s not indexed in $T$ (i.e., the remaining $r-t$ base agents) receive either $\bigcup_{j=1}^6 e^i_j$ or $\bigcup_{j=8}^{13} e^i_j$. Note that, for each $i \in T$, the base agent $z_i$ obtains a value of $1/4$ in $\mathcal{I}$ and the base agents not indexed in $T$ obtain a value of $1/2$. Indeed, in comparison to the YES case, here the base agents indexed in $T$ and the clause agents indexed in $P$ receive lower values. The following claim provides a lower bound on the number of such agents; the proof of this claim is deferred to Appendix~\ref{section:proof-of-claim-2}.
\begin{restatable}{claim}{ClaimBadAgents}
\label{claim:badAgents}
With $t$ and $p$ as defined above, we have $4t + p \geq \alpha m$.
\end{restatable}

Note that under the allocation $\mathcal{I}$, the auxiliary interval $G$ is divided amongst $p$ clause agents (indexed in $P$) and the auxiliary agent $d$.\footnote{$I_d \subseteq G$ follows from the fact that $\NSW(\mathcal{I}) >0$ (property (iii) in Claim~\ref{claim:structure}), i.e., $v_d(I_d) >0$, and the auxiliary agent $d$ has nonzero valuation (density) only within $G$.} The following claim upper bounds the contribution of these agents to $NSW(\mathcal{I})$. 

\begin{restatable}{claim}{ClaimDummyNode}
\label{claim:dummyNode}
 In the above setting, we have $v_d(I_d) \ \prod_{j \in P} v_{a_j}(I_{a_j})  \leq \left(\frac{2}{3} \ \frac{1}{p+1} \right)^p \frac{1}{p+1}$.
\end{restatable}

The proof of Claim~\ref{claim:dummyNode} follows from basic calculations and appears in Appendix~\ref{section:proof-of-claim-3}. The above-mentioned observations and Claim \ref{claim:dummyNode} give us the following upper bound on $NSW_N$ (the optimal Nash social welfare in $\mathcal{C}(\phi)$):
\begin{align}
\label{equation:NOcase}
   NSW_N = \NSW(\mathcal{I}) \leq \left(1^{2r} \frac{1}{4^t}  \frac{1}{2^{r-t}} \frac{1}{3^{m-p}} \left( \frac{2}{3} \frac{1} {p+1}\right)^p \frac{1}{p+1} \right)^{\frac{1}{3r+m+1}}
\end{align}

Write $\tau \coloneqq \left(1^{2r} \frac{1}{2^r} \frac{1}{3^m} 1 \right)^{\frac{1}{3r+m+1}}$ and $c(\alpha) \coloneqq 2^{-\alpha/44}$. Using Claim \ref{claim:badAgents}, equation \eqref{equation:NOcase}, and the observation that $r \leq 3m$, we obtain the following inequality: $NSW_N \leq c(\alpha) \tau$; see Appendix \ref{section:worst-subsection-title-ever} for details. Also, equation \eqref{equation:YEScase} gives us $NSW_Y \geq \tau$. Therefore, using a $\left(1/c(\alpha)\right)$-approximation algorithm for the Nash social welfare maximization problem, one can distinguish between the YES and the NO cases of $\textrm{Gap 3-SAT-5}$. Since the latter problem is $\NP$-hard, such an approximation algorithm does not exist, unless ${\rm P} = {\rm NP}$. That is, we obtain $\left(1/c(\alpha)\right)$-inapproximability of maximizing Nash social welfare and, hence, Theorem~\ref{theorem:nswhardness} holds.

\section{Hardness of Maximizing $\rho$-Mean Welfare} \label{section:rhomeanhardness}

In this section we prove that, for any fixed $\rho \in (0,1]$, it is $\APX$-hard to find an allocation that maximizes the $\rho$-mean welfare. 

Recall that, in a cake-division instance $\langle [n], \{v_a \}_{a \in [n] } \rangle$, the $\rho$-mean welfare of an allocation $\mathcal{I} =\{I_1, \ldots, I_n \}$ is defined as $\mathrm{M}_\rho(\mathcal{I}) \coloneqq \left( \frac{1}{n} \sum_{a=1}^n [v_a(I_a)]^\rho \right)^{1/\rho}$. Also, as mentioned previously, the range $\rho \in (0,1]$ constitutes a family of functions that captures both average social welfare and Nash social welfare: $\rho=1$ gives us the arithmetic mean (social welfare) and, as $\rho$ tends to zero, the limit of $\mathrm{M}_\rho$ is equal to the geometric mean (the Nash social welfare).



In contrast to the $\APX$-hardness result for maximizing Nash social welfare (Theorem \ref{theorem:nswhardness}), our reduction for $\rho$-mean welfare does not yield a cake-division instance with normalized valuations. That is, here we will consider cake-division instances wherein for some agents the total value of the cake is not necessarily equal to one. However, for the $\rho = 1$ case (i.e., for average social welfare), one can obtain $\APX$-hardness with normalized valuations--this entails a direct adaptation of the construction developed in Section~\ref{section:gadget}.

The main result of this section is as follows

\begin{theorem}\label{theorem:rhohardness}
For each fixed $\rho \in (0,1]$, it is $\APX$-hard to find allocations that maximizes the $\rho$-mean welfare in cake-division instances with piecewise-constant valuations. {Here, the agents' valuations are not necessarily normalized over the cake.} 
\end{theorem}

\begin{proof}
We prove this theorem by developing a gap-preserving reduction from $\textrm{Gap 3-SAT-5}$ to the $\rho$-mean welfare maximization problem. Here, the construction of a cake-division instance (from a $\textrm{Gap 3-SAT-5}$ instance) is quite similar to the one detailed in Section~\ref{section:gadget}.  Hence, we will primarily present the parts where the two constructions differ and omit the commonalities. 
	
We start with an instance of the $\textrm{Gap 3-SAT-5}$ problem, comprising of a Boolean formula $\phi$ (in conjunctive normal form), such that---for a fixed constant $\alpha \in (0,1)$---it is $\NP$-hard to distinguish between the YES case ($\phi$ is satisfiable) and the NO case (no assignment satisfies more than $(1-\alpha)$ fraction of the clauses). As before, $\phi$ comprises of a conjunction of a set of clauses, $\textrm{C} = \{ C_1, C_2, \ldots, C_m \}$, defined over Boolean variables $\textrm{X}=\{x_1, \ldots, x_r \}$.

From formula $\phi$, we construct a cake-division instance $\mathcal{C}(\phi)$ which achieves the necessary gap-property for $\rho$-mean welfare. 
 We follow the construction detailed in Section~\ref{section:gadget} with the following changes 
 \begin{enumerate}
\item  Remove the auxiliary interval $G$ and the auxiliary agent $d$ from the construction.\footnote{Here, we do not impose normalization, i.e., do not enforce each agent's value for the entire cake to be equal to one. Hence, in the current context, the auxiliary interval and auxiliary agent serve no purpose.} Hence, the cake $[0, 14r]$ is obtained by concatenating the $r$ intervals, $H_1, \ldots, H_r$. 
\item Modify the valuations (as defined in Section \ref{section:gadget}) for each of the remaining $3r+m$ agents as follows: for an agent $a$ and subinterval $e^i_k$ (in interval $H_i$), let $v_a(e^i_k)$ be the value considered in Section~\ref{section:gadget}, e.g., $v_{z_i} ( e^i_1) = 1/4$. Here, we modify $v_a(e^i_k) \leftarrow \left( v_a(e^i_k) \right)^{1/\rho}$, for all $a \in [n]$, $i\in [r]$, and $k \in [14]$. For instance, we now have $v_{z_i}\left( \cup_{k=1}^6 e^i_k\right) = \frac{2}{4^{1/\rho}}$. Whenever an agent has a nonzero value for a subinterval, we keep the valuation (density) to be constant, i.e., we ensure that the resulting cake-division instance consists of piecewise-constant valuations. 
\end{enumerate}
    
    In summary, we construct in polynomial-time a cake-division instance $\C(\phi)$ with $2r$ separator agents, $r$ base agents and $m$ clause agents; hence $n = 3r+m$. Observe that, with slight modifications in the arguments, Claim \ref{claim:structure} continues to hold for the current construction. 
    
To prove the theorem, we first consider the YES case, i.e., the case wherein $\phi$ is satisfiable. In this setting, there exists an assignment $f:X \to \{0,1\}$ of the variables that satisfies all clauses in $\phi$. Let $M^Y_{\rho}$ denote the optimal $\rho$-mean welfare in the cake-division instance $\mathcal{C}({\phi})$. We follow a similar procedure as discussed in Theorem~\ref{theorem:nswhardness}, to obtain the following upper bound:
    \begin{equation}
    \label{equation:YESrho}
    M^Y_{\rho} \geq  \Bigg(\frac{1}{3r+m} \Big[ 2r + r  \frac{2^{\rho}}{4} + m \frac{1}{3}  \Big] \Bigg)^{\frac{1}{\rho}}
    \end{equation}

Next, we consider the NO case, i.e., under any assignment of the variables, at least $\alpha m$ clauses in $\phi$ remain unsatisfied. Let $\mathcal{I}=\{I_a\}_a$ be a $\rho$-mean optimal allocation (of $\mathcal{C}({\phi})$) which satisfies the properties stated in Claim~\ref{claim:structure}. 
    
Furthermore, note that Claim~\ref{claim:badAgents} continues to hold in the present context with the following modification: denote by $P'$ the set of clause agents that receive a value of zero under the optimal allocation $\I$, i.e., $P' \coloneqq \{ j \in [m] \mid v_{a_j} (I_{a_j}) = 0 \}$. Write $p' \coloneqq |P'|$.\footnote{In this construction, a clause agent can potentially receive a value of $0$ even under an optimal allocation. This follows from the observation that in the current setup there does not exist an auxiliary interval $G$ that can be shared (as a fallback) among the clause agents to ensure that they receive strictly positive values.} The set $T$ remains unchanged (as defined in Section \ref{section:hardnessproof}), i.e., $T \coloneqq \{i \in [r] \mid I_{z_i} = e^i_1 \text{ or } I_{z_i} = e^i_8\}$. The arguments presented in Claim \ref{claim:badAgents} gives us the inequality $4t+p' \geq \alpha m$.
    
   Let $M_{\rho}^N$ denote the optimal $\rho$-mean welfare in the cake-division instance $\mathcal{C}({\phi})$, under the assumption that $\phi$ satisfies the NO case. We once again follow a similar procedure as discussed in Theorem \ref{theorem:nswhardness} to obtain the following upper bound:
   \begin{equation}
   \label{equation:NOrho}
   M_{\rho}^N \leq \Bigg(\frac{1}{3r+m} \Big[ 2r + (r-t)  \ \frac{2^{\rho}}{4} + t  \ \frac{1}{4} + (m-p') \  \frac{1}{3} + p' \cdot 0 \Big] \Bigg)^{\frac{1}{\rho}}
    \end{equation}
    
   Let $\tau' \coloneqq \Bigg(\frac{1}{3r+m} \Big[ 2r + r  \frac{2^{\rho}}{4} + m \frac{1}{3}  \Big] \Bigg)^{\frac{1}{\rho}}$; see equation \eqref{equation:YESrho}.
  For showing that the reduction is gap preserving, it suffices to bound the ratio $\frac{M_{\rho}^N}{\tau'}$ by a constant, say $c(\alpha, \rho)$.  Towards this end, we analyze the $\rho$th power of the ratio $\frac{M_{\rho}^N}{\tau'}$:
      
\begingroup      
\allowdisplaybreaks
     
      	\begin{align}
      \Bigg(\frac{M_{\rho}^N}{\tau'}\Bigg)^{\rho} & \leq \ \frac{2r + (r-t)  \frac{2^{\rho}}{4} + t  \frac{1}{4} + (m-p')  \frac{1}{3}}{ 2r + r  \frac{2^{\rho}}{4} + m \frac{1}{3} } \nonumber \\
      & = 1 - \bigg[  \frac{ \frac{t}{4} (2^{\rho} -1) + \frac{p'}{3}}{ 2r + r  \frac{2^{\rho}}{4} + \frac{m}{3} } \bigg]  \nonumber \\
      & \leq  1 - \bigg[  \frac{ \frac{(\alpha m - p')}{16} (2^{\rho} -1) + \frac{p'}{3}}{ 6m + 3m \frac{2^{\rho}}{4} + \frac{m}{3} } \bigg]    \tag{$4t+p' \geq \alpha m$ from Claim~\ref{claim:badAgents} and $r \leq 3m$}   \\     
      & \leq  1 - \bigg[  \frac{ \frac{\alpha m}{16} (2^{\rho} -1) + p' (\frac{1}{3} - \frac{2^{\rho} -1}{16}) }{6m + 3m  \frac{2^{\rho}}{4} + \frac{m}{3}} \bigg]  \nonumber \\
      & \leq  1 - \bigg[  \frac{ \frac{\alpha m}{16} (2^{\rho} -1) }{m (6 + 3 \frac{2^{\rho}}{4} + \frac{1}{3})} \bigg]  \tag{Since  $\frac{1}{3} - \frac{2^{\rho} -1}{16} > 0$ for all $\rho \in (0,1)$}  \\
      & \leq  1 - \bigg[  \frac{ \frac{\alpha m}{16} (2^{\rho} -1) }{m (6 + 3 \frac{2^{\rho}}{4} + \frac{1}{3})} \bigg]   \nonumber \\
      & = 1 -  \frac{ 3\alpha (2 ^{\rho} -1)}{4 (18 + 9 \ 2^{\rho} + 4)} \label{ineq:rho-mess}
      \end{align}
\endgroup

With constant $c(\alpha, \rho) := \bigg(1 -  \frac{ 3\alpha (2 ^{\rho} -1)}{4 (18 + 9 \ 2^{\rho} + 4)}\bigg)^{1/\rho}$, inequality \eqref{ineq:rho-mess} can be restated as $\frac{M_{\rho}^N}{\tau'} \leq c(\alpha, \rho) $. Using this bound and equation \eqref{equation:YESrho}, we obtain a (constant) multiplicative gap of $c(\alpha, \rho)$ between $M_{\rho}^N$ and $M_{\rho}^Y$. 

Therefore, using a $\left(1/c(\alpha, \rho)\right)$-approximation algorithm for the $\rho$-mean welfare maximization problem, one can distinguish between the YES and the NO cases of $\textrm{Gap 3-SAT-5}$. Since the latter problem is $\NP$-hard, such an approximation algorithm does not exist, unless ${\rm P} = {\rm NP}$. That is, we obtain $\left(1/c(\alpha, \rho)\right)$-inapproximability of maximizing $\rho$-mean welfare (with unnormalized valuations) and, hence, the theorem holds. 
\end{proof}

\section{Approximating Nash Social Welfare with a Constant Number of Agents}
\label{appendix:const-num-agents-nsw}

This section shows that an $\alpha$-approximate solution (with $\alpha>1$) for the Nash social welfare maximization problem (among $n$ agents) can be computed in time $\left(\frac{n}{\log \alpha} \right)^{\mathcal{O} \left( n \right)}$. In particular, for a constant number of agents and constant $\alpha >1$, we obtain a polynomial-time $\alpha$-approximation algorithm.  

The high-level idea here is to first discretize the set of possible values that the agents can obtain in a Nash optimal allocation and then find a near-optimal solution via an exhaustive search over the discretized set. 

\begin{theorem} \label{thereom:constant-num-agents-nsw}
For cake-division instances $\left\langle [n], \{v_a\}_{a \in [n]} \right\rangle$, with piecewise-constant valuations, and $\alpha >1$, there exists an algorithm that---in time $\left(\frac{n}{\log \alpha} \right)^{\mathcal{O} \left( n \right)}$---finds an  $\alpha$-approximation to the Nash social welfare maximization problem. 
\end{theorem}

\begin{proof}
	Let $\mathcal{I}^* = \{I^*_1, \ldots, I^*_n\}$ denote a Nash optimal allocation for the given cake-division instance. Here, interval $I^*_a$ is assigned to agent $a \in [n]$. Also, let permutation $\sigma^* \in \mathbb{S}_n$ denote the order in which the intervals $I^*_a$s are assigned across the cake $[0,1]$, i.e., under allocation $\mathcal{I}^*$, the interval assigned to agent $\sigma^*(1)$ is at the left end of the cake and, continuing on, agent $\sigma^*(n)$ receives the right-most interval. 
	
	To begin with, note that the Nash social welfare of a proportional allocation is at least $1/n$ and, hence, $\NSW(\mathcal{I}^*) \geq 1/n$. This implies that $v_a(I^*_a) \geq \frac{1}{n^n}$ for all agents $a \in [n]$. Therefore, for each agent $a \in [n]$, we have $v_a(I^*_a) \in \left[ \frac{1}{n^n}, 1 \right]$.\footnote{Recall that $v_a([0,1])=1$ for all $a\in [n]$.} We discretize this range of values. Specifically, with $\alpha > 1$, write $G \coloneqq \left\{ \frac{1}{n^n}, \frac{\alpha}{n^n}, \frac{\alpha^2}{n^n}, \frac{\alpha^3}{n^n}, \ldots, 1 \right\}$ and note that the cardinality of this set satisfies $|G| \leq \frac{n \log n}{\log \alpha}$. 
	
	Next we observe that that there necessarily exists an allocation $\widehat{\mathcal{I}} =\{\widehat{I}_1, \ldots, \widehat{I}_n \}$ which satisfies $\NSW(\widehat{\mathcal{I}}) \geq \frac{1}{\alpha} \NSW(\mathcal{I}^*)$ and $v_a(\widehat{I}_a) \in G$ for all $a \in [n]$. In particular, consider the valuation vector $\left( v_a(I^*_a) \right)_{a \in [n]}$ and round down each component to the closest value in set $G$. That is, for each agent $a \in [n]$, set $\widehat{v}_a \coloneqq \max \left\{  \frac{\alpha^i}{n^n} \ : \  0 \leq i \leq |G| \text{ and } \frac{\alpha^i}{n^n} \leq  v_a(I^*_a) \right\}$. This definition ensures $\frac{1}{\alpha} v_a(I^*_a) \leq \widehat{v}_a \leq v_a(I^*_a)$ for each agent $a \in [n]$.  Furthermore, note that following the order $\sigma^*(1), \sigma^*(2), \ldots, \sigma^*(n)$, we can find the desired allocation $\widehat{\mathcal{I}}$, i.e., establish the existence of an allocation wherein each agent $a$ receives an interval of value $\widehat{v}_a$:\footnote{The agent $\sigma^*(n)$ can potentially receive an interval of higher value, however this corner case does not affect the developed arguments.} for agent $\sigma^*(1)$ consider the cut point $c_1 \in [0,1]$ such that $v_{\sigma^*(1)} [ 0, c_1] = \widehat{v}_{\sigma^*(1)}$. Then, we simply repeat this procedure for the remaining agents (in order $\sigma^*(2), \ldots, \sigma^*(n)$) and over the remaining cake $[c_1, 1]$.  

The existence of $\widehat{\mathcal{I}}$ implies that, by enumerating over all possible valuation vectors with components in $G$ and all $n!$ permutations, we can find an allocation with Nash social welfare at least $\frac{1}{\alpha}$ times the optimal. The number of valuation vectors induced by values in $G$ is at most $|G|^n$, hence the time complexity of this exhaustive search is $\mathcal{O} \left( |G|^n \ n! \right)$.  Since $|G| \leq \frac{n \log n}{\log \alpha}$, we obtain, for the Nash social welfare maximization problem, an $\alpha$-approximation algorithm that runs in time $\left(\frac{n}{\log \alpha} \right)^{\mathcal{O} \left( n \right)}$.
\end{proof}	

\section{Price of Envy-Freeness}
\label{section:priceEF}
In this section we establish an upper bound on the price of envy-freeness for $\rho$-mean welfare. The price of envy-freeness with respect to a welfare objective (such as social welfare or, more generally, $\rho$-mean welfare)  is defined as the ratio of the optimal value of the welfare objective among all possible divisions and the optimal value of the welfare objective among all envy-free divisions (i.e., the optimal value achieved under the envy-freeness constraint). Theorem~\ref{theorem:priceEF} establishes an upper bound on the price of envy-freeness achieved by $\alpha$-$\EF$ allocations for $\rho$-mean welfare, with $\alpha \geq 1$ and $\rho \in (0,1]$. 

Note that, along these lines, Theorem~\ref{theorem:ef-nsw} implies an upper bound of $2\alpha$ on the price of envy-freeness obtained by $\alpha$-EF allocations for Nash social welfare (i.e., for the $\rho \to 0$ case). 

Furthermore, as a direct consequence of Theorem~\ref{theorem:priceEF} and Theorem~\ref{theorem:2ef}, we get that $\ALG$ provides an $ \mathcal{O}(2^{\frac{1}{\rho}}\  n^{\frac{\rho}{\rho+1}})$-approximate solution to maximizing $\rho$-mean welfare for $\rho \in (0,1]$. For constant $\rho \in (0,1]$, this approximation factor is weaker than the one obtained in Theorem~\ref{theorem:interval-approx}. However, the result obtained here is universal in the sense that a single allocation achieves the guarantee for all $\rho \in (0,1]$.  


\begin{restatable}{theorem}{TheoremPriceofEF}
 \label{theorem:priceEF}
Let $\mathcal{I}$ be an $\alpha$-approximately envy-free allocation in a cake-division instance $\mathcal{C}=\{[n], \{v_a\}_a \}$, with $\alpha \geq 1$. Then, for any $\rho \in (0,1]$, $\mathcal{I}$ provides a $ \left( 2 \alpha \ 2^{\frac{1}{\rho}}\  n^{\frac{\rho}{\rho+1}} \right)$-approximate solution to maximizing $\rho$-mean welfare in $\mathcal{C}$.  
\end{restatable}

\begin{proof}
Fix $\rho \in (0,1]$ and let $\mathcal{I}^* = \{I^*_1, \ldots, I^*_n\}$ be an optimal allocation with respect to the $\rho$-mean welfare. For $a \in [n]$, write $K_a$ to denote the set of intervals in the $\alpha$-$\EF$ allocation $\mathcal{I} = \{I_1, \ldots, I_n\}$ that intersect with $I^*_a$, i.e., $K_a \coloneqq \{ I_b \in \mathcal{I} \mid I_b \cap I^*_a \neq \emptyset \}$.\footnote{Here, we follow the previously mentioned convention that mandates two intervals to be disjoint, if they intersect exactly at an endpoint.} Let $k_a$ denote the cardinality of this set, $k_a \coloneqq |K_a|$. 

Since $\mathcal{I}$ is an allocation, we have $\cup_a I_a = [0,1]$. Therefore, for all agents $a \in [n]$, interval $I^*_a$ is covered by the union of intervals in $K_a$; $I^*_a \subseteq \cup_{I \in K_a} I$. Also, the fact that $\mathcal{I}$ is $\alpha$-approximately envy-free implies $\alpha v_a(I_a) \geq v_a(I_b)$ for all $I_b \in K_a$. Summing these inequalities and using the containment $I^*_a \subseteq \cup_{I \in K_a} I$ gives us\footnote{Recall that the valuations are sigma additive.} $\alpha k_a \ v_a(I_a) \geq v_a(I^*_a)$ for all $a\in [n]$.  

Furthermore, the normalization of valuations ($v_a([0,1]) = 1$), gives us $\min\{ 1  , \alpha k_a \ v_a(I_a) \} \geq v_a(I^*_a)$ for all $a\in [n]$.
    
Define a subset of players, $\mathcal{H} \coloneqq \{a \in [n] \mid k_a \geq 2n^{\frac{\rho}{\rho +1}}\}$. Note that $2n^{\frac{\rho}{\rho +1}}|\mathcal{H}| \leq \sum_{a \in \mathcal{H}} k_a \leq \sum_{a=1}^n k_a \leq 2n$; here, the last inequality follows from a counting argument, which is detailed in the proof of Theorem~\ref{theorem:ef-nsw}. These inequalities imply the following upper bound on the size of $\mathcal{H}$: $|\mathcal{H}| \leq n^{1 - \frac{\rho}{\rho +1}}$. 

Considering the optimal $\rho$-mean welfare value
    \begin{align*}
    \mathrm{M}_\rho(\mathcal{I^*}) &= \left( \frac{1}{n} \sum_{a=1}^n [v_a(I^*_a)]^\rho \right)^{1/\rho} \\
    &= \left( \frac{1}{n} \sum_{a \in \mathcal{H}} [v_a(I^*_a)]^\rho + \frac{1}{n} \sum_{a \notin \mathcal{H}} [v_a(I^*_a)]^\rho \right)^{1/\rho} \\
    & \leq \left( \frac{1}{n} \sum_{a \in \mathcal{H}} 1 + \frac{1}{n} \sum_{a \notin \mathcal{H}} [\alpha k_a \ v_a(I_a)]^\rho \right)^{1/\rho} \tag{since $v_a(I^*_a) \leq \min\{1, \alpha k_a \ v_a(I_a)\}$} \\
    & \leq \left( \frac{1}{n}\ |\mathcal{H}| + \frac{1}{n} \sum_{a \notin \mathcal{H}} [2\alpha \ n^{\frac{\rho}{\rho +1}} \ v_a(I_a)]^\rho \right)^{1/\rho} \tag{since $k_a \leq 2n^{\frac{\rho}{\rho +1}}$ for all $a \notin \mathcal{H}$}\\
    & \leq \left( \frac{1}{n^{\frac{\rho}{\rho +1}}} + \frac{1}{n}\ \left(2\alpha \ n^{\frac{\rho}{\rho +1}} \right)^\rho \sum_{a \notin \mathcal{H}} [v_a(I_a)]^\rho \right)^{1/\rho} \tag{since $|\mathcal{H}| \leq n^{1 - \frac{\rho}{\rho +1}}$}
    \end{align*}
    Note that $\frac{1}{\rho} \geq 1$, hence $f(x) \coloneqq x^{1/\rho}$ is a convex function on $\mathbb{R}_+$. The convexity of this function gives us $\left( \frac{a + b}{2} \right)^{1/\rho} \leq \frac{1}{2} a^{1 / \rho} + \frac{1}{2} b^{1 / \rho}$. Simplifying we get $(a + b)^{1/\rho} \leq 2^{\frac{1}{\rho} - 1} (a^{1/\rho} + b^{1/\rho})$.
    
    We combine this observation with the bound obtained for the optimal $\rho$-mean welfare
    \begin{align}
    \mathrm{M}_\rho(\mathcal{I^*}) &\leq \left( \frac{1}{n^{\frac{\rho}{\rho +1}}} + \frac{1}{n}\ \left(2\alpha \  n^{\frac{\rho}{\rho +1}} \right)^\rho \sum_{a \notin \mathcal{H}} [v_a(I_a)]^\rho \right)^{1/\rho} \nonumber \\
    & \leq 2^{\frac{1}{\rho} - 1} \left( \frac{1}{n^{\frac{1}{\rho +1}}} + \left( 2\alpha n^{\frac{\rho}{\rho +1}} \right) \left( \frac{1}{n} \sum_{a \notin \mathcal{H}} [v_a(I_a)]^\rho \right)^{1/\rho} \right)  \label{ineq:opt-rho}
    \end{align}
    
   Since $\mathcal{I}$ is an $\alpha$-$\EF$ allocation, we have $ \alpha n \ v_a(I_a) \geq \sum_{b=1}^n v_a(I_b) = v_a([0,1]) = 1$. That is, $v_a(I_a) \geq \frac{1}{\alpha n}$ for all $a \in [n]$. Therefore, the $\rho$-mean welfare of $\mathcal{I}$ is at least $\frac{1}{\alpha n}$, i.e., $ \mathrm{M}_\rho(\mathcal{I}) \geq \frac{1}{\alpha n}$. We can hence write $\frac{1}{n^{\frac{1}{\rho +1}}} \leq \alpha n^{\frac{\rho}{\rho +1}} \mathrm{M}_\rho(\mathcal{I})$. 
   
   This inequality and equation \eqref{ineq:opt-rho} provide the desired upper bound on the optimal $\rho$-mean welfare
    \begin{align*}
    \mathrm{M}_\rho(\mathcal{I^*}) &\leq 2^{\frac{1}{\rho} - 1} \left( \frac{1}{n^{\frac{1}{\rho +1}}} + \left(2\alpha n^{\frac{\rho}{\rho +1}} \right) \left( \frac{1}{n} \sum_{a \notin \mathcal{H}} [v_a(I_a)]^\rho \right)^{1/\rho} \right)\\
    & \leq 2^{\frac{1}{\rho} - 1} \left( \alpha \ n^{\frac{\rho}{\rho +1}} \mathrm{M}_\rho(\mathcal{I}) + 2 \alpha \ n^{\frac{\rho}{\rho +1}} \mathrm{M}_\rho(\mathcal{I}) \right)\\
    & \leq  2 \alpha \ 2^{\frac{1}{\rho}} \ n^{\frac{\rho}{\rho+1}} \ \mathrm{M}_\rho(\mathcal{I})
    \end{align*}
    \end{proof}

\section{Conclusions and Future Work}
The current work studies cake-cutting from an algorithmic perspective and obtains approximation guarantees for multiple, well-studied notions of fairness and efficiency. In particular, we develop an efficient algorithm that computes $(2 + o(1))$-approximately envy-free allocations and, simultaneously, provides a $(3 + o(1))$-approximation to Nash social welfare. We complement this algorithmic result for Nash social welfare by proving that, in the cake-cutting context, maximizing this objective is $\APX$-hard. Developing hardness results for (approximate) envy-freeness remains an interesting open problem.\footnote{Since an envy-free cake division always exists, the hardness results here will be in terms of complexity classes contained in ${\rm TFNP}$.} Notably, the result of Deng et al.~\cite{deng2012algorithmic} shows that envy-free cake division (with connected pieces) is ${\rm PPAD}$-hard, but this negative result holds under ordinal valuations--in this setup the preferences of each agent is specified via an explicit circuit which, given an allocation, identifies the agent's most preferred piece. Therefore, in and of itself, the result of Deng et al.~\cite{deng2012algorithmic} does not imply that envy-free cake division under cardinal valuations is ${\rm PPAD}$-hard; complementarily, this result does not rule out an FPTAS for the contiguous-pieces version of envy-free cake-cutting under, say, piecewise-constant valuations.

Our approximation guarantee for $\rho$-mean welfare degrades as $\rho$ tends to zero. Indeed, it does not match the approximation ratio achieved specifically for Nash social welfare. Tightening this gap is another interesting direction for future work. Computational results for maximizing $\rho$-mean welfare, with $\rho < 0$, will also be interesting. The $\rho \to -\infty$ case is particularly relevant, since it corresponds to egalitarian welfare, i.e., to the max-min (Santa Claus) objective. The work of Aumann et al.~\cite{aumann2013computing} proves that, in the cake-division framework, it is ${\rm NP}$-hard to approximate egalitarian welfare within a factor of two. However, it remains open whether this problem admits a nontrivial approximation algorithm.

\section*{Acknowledgements}
Siddharth Barman gratefully acknowledges the support of a Ramanujan Fellowship (SERB - {SB/S2/RJN-128/2015}) and a Pratiksha Trust Young Investigator Award.

\bibliographystyle{alpha}
\bibliography{references}

\appendix

\section{Omitted Proofs from Section \ref{section:hardnessproof}}
\label{section:nswhardness_appendix}

\subsection{Proof of Claim~\ref{claim:structure}} 
\label{section:proof-of-claim-1}

\ClaimStructuralProperties*

\begin{proof} 
We will show that, independent of the underlying case for $\phi$ (YES or NO), the constructed cake-division instance $\mathcal{C}(\phi)$ admits a Nash optimal allocation, $\mathcal{I} = \{I_a\}_a$ which satisfies the properties mentioned in the claim. 

Property (iii) can be established by noting that the following partial allocation has positive Nash social welfare: for all $i \in [r]$, allocate subinterval $e^i_1$ to agent $z_i$, subinterval $e^i_7$ to agent $s_i$, and $e^i_{14}$ to agent $s'_i$. Say literal $x_i$ appears in the clause $C_j \in \textrm{C}$ and, for $q \in \{1,2,3, 4\}$, $C_j$ is the $q$th positive occurrence of $x_i$. In this case, we allocate subinterval $e^i_{1+q}$ to clause agent $a_j$. In the complementary setting, if literal $\overline{x}_i$ appears in the clause $C_j$ (with $C_j$ being the $q$th negative occurrence of variable $x_i$), then we allocate the subinterval $e^i_{8+q}$ to $a_j$. Note that, in both cases, the clause agent $a_j$ obtains a positive value under the partial allocation. Finally, assign interval $G$ to the auxiliary agent $d$. These allocations lead to a well-defined partial allocation wherein every agent receives a positive value, i.e., the partial allocation has positive Nash social welfare. Since the valuations of the agents are nonnegative, this partial allocation can be extended to a full allocation without reducing the value of any agent. Hence, every Nash optimal allocation satisfies property (iii).

Let $\mathcal{J}=\{J_a\}_a$ be a Nash optimal allocation in the cake-division instance $C(\phi)$. Using $\mathcal{J}$, we will define a different Nash optimal allocation, $\mathcal{I}=\{I_a\}_a$, that satisfies properties (i) and (ii). 

By property (iii), we have that $v_a(J_a) > 0$ for all agents $a$. Hence, for all $i \in [r]$, we have $J_{s_i} \cap e^i_7 \neq \emptyset$ and $J_{s'_i} \cap e^i_{14} \neq \emptyset$---agent $s_i$ ($s'_i$) has nonzero valuation (density) only within the subinterval $e^i_7$ ($e^i_{14}$). This observation implies that for all $i \in [r]$, the interval assigned to the base agent $z_i$ (i.e., $J_{z_i}$) satisfies one of the following cases
\begin{enumerate}
    \item  $J_{z_i}$ intersects with exactly one of $e^i_1$ and $e^i_6$. In this case, we set $I_{z_i} \coloneqq e^i_1$.
    \item  $J_{z_i}$ intersects with both $e^i_1$ and $e^i_6$. In this case, we set $I_{z_i} \coloneqq \bigcup_{k=1}^6 e^i_k$.
    \item  $J_{z_i}$ intersects with exactly one of $e^i_8$ and $e^i_{13}$. In this case, we set $I_{z_i} \coloneqq e^i_8$.
    \item  $J_{z_i}$ intersects with both $e^i_8$ and $e^i_{13}$. In this case, we set $I_{z_i} \coloneqq \bigcup_{k=8}^{13} e^i_k$.
\end{enumerate}

Hence, case-wise, we have defined the interval assigned to base agent $z_i$ (i.e., defined $I_{z_i}$) in the allocation $\mathcal{I}$. In addition, for all $i \in [r]$, set $I_{s_i} \coloneqq e^i_7$ and $I_{s'_i} \coloneqq e^i_{14}$. One can verify that, for all $i \in [r]$, we have $v_{z_i}(J_{z_i}) \leq v_{z_i}(I_{z_i})$, $v_{s_i}(J_{s_i}) \leq v_{s_i}(I_{s_i})$ and $v_{s'_i}(J_{s'_i}) \leq v_{s'_i}(I_{s'_i})$. That is, these modifications in the Nash optimal allocation $\mathcal{J}$ do not reduce the values obtained by the base agents and the separator agents. 

Let $B$ denote the subset of the cake $[0,14r+1] $ that has been already allocated to the separator and base agents under $\mathcal{I}$, i.e., $B \coloneqq \bigcup_{i=1}^{r} (I_{z_i} \cup I_{s_i} \cup I_{s'_i}$). Write $B^{\mathsf{c}}$ to denote $[0,14r+1] \setminus B$. For each clause agent $a_j$ (with $j \in [m]$)  set interval $I_{a_j} \coloneqq J_{a_j} \cap B^{\mathsf{c}}$. {Note that $J_{a_j} \cap B^{\mathsf{c}}$ is indeed an interval.} From the construction of $\mathcal{C}(\phi)$, we know that clause agents do not value any interval contained in $B$. Hence, $v_{a_j}(I_{a_j}) = v_{a_j}(J_{a_j})$ for all $j \in [m]$. Similarly, for the auxiliary agent, set $I_d \coloneqq J_d \cap G$ (for the auxiliary interval $G$, we have $ G \subset B^{\mathsf{c}}$) and note that $v_d(I_d) = v_d(J_d)$. Therefore, $NSW(\mathcal{I}) \geq NSW(\mathcal{J}$) and, hence, $\mathcal{I}$ is a Nash optimal allocation that satisfies all the properties stated in the claim. 
\end{proof}

\subsection{Proof of Claim \ref{claim:badAgents}}
\label{section:proof-of-claim-2}

We consider this claim when the given Boolean formula $\phi$ satisfies the NO case. Here, $\mathcal{I}=\{I_a\}_a$ is a Nash optimal allocation in the cake-division instance $\mathcal{C}(\phi)$ and $\mathcal{I}$ satisfies the properties in Claim \ref{claim:structure}. Also, recall that $T \coloneqq \{i \in [r] \mid I_{z_i} = e^i_1 \text{ or } I_{z_i} = e^i_8\}$ and $P \coloneqq \{j \in [m] \mid I_{a_j} \subset G\}$, with $t \coloneqq |T|$ and $p \coloneqq |P|$.

\ClaimBadAgents*
 
\begin{proof} 
 Using $\mathcal{I}$, we will define a different allocation, denoted by $\mathcal{L}=\{L_a\}_a$, which will be used to extract an assignment to the Boolean variables in $\phi$. 

For all $i \in [r]$, set $L_{s_i} \coloneqq I_{s_i}$ and $L_{s'_i} \coloneqq I_{s'_i}$. In addition, for $i \notin T$, set $L_{z_i} = I_{z_i}$. On the other hand, for $i \in T$, set $L_{z_i} \coloneqq \bigcup_{k=1}^6 e^i_k$ if $I_{z_i} = e^i_1$ and set $L_{z_i} \coloneqq \bigcup_{k=8}^{13} e^i_k$ if $I_{z_i} = e^i_8$. That is, for every base agent indexed in $T$, we enlarge the interval $I_{z_i}$ to cover the entire positive part of $H_i$ or the entire negative part of $H_i$. Also, for the auxiliary agent $d$, assign $L_d \coloneqq G$. 

Using this, we consider the set of clause agents $a_j$s whose initial allocations, $I_{a_j}$s, are now contained in the enlarged allocations of the base agents; specifically, write $S \coloneqq \{j \in [m] \mid \text{ there exists } i \in T \text{ such that } I_{a_j} \subset L_{z_i}\}$. Note that, for every $j \in S$, we have an intersection $I_{a_j} \cap L_{z_i} \neq \emptyset$, for some base agent $z_i$. Also, following the definition of $P$, we get $I_{a_j} \cap L_d \neq \emptyset$, for all $j \in P$. To resolve these intersections, for each $j \in S \cup P$, we set $L_{a_j} \coloneqq \emptyset$. Complementarily, for $j \notin S \cup P$, set $L_{a_j} \coloneqq I_{a_j}$. It can be verified that  $\mathcal{L}=\{L_a\}_a$ is a well-defined allocation and under it each base agent $z_i$ receives either the entire positive part or entire the negative part of $H_i$. 

Note that the clause agents indexed in $S$ are forced to receive an empty allocation and the number of such agents (i.e., $|S|$) is at most $4|T| = 4t$. This upper bound follows from the fact that, for any $i \in [r]$, the containment $I_{a_j} \subset \cup_{k=2}^5 e^i_k$  holds for at most four clause agents $a_j$s--recall that $v_{a_j}(I_{a_j}) > 0$ for all $j \in [m]$ (property (iii) in Claim~\ref{claim:structure}) and, hence, if this containment holds for a $j \in [m]$, then it must be the case that literal $x_i$ occurs in clause $C_j$. The number of occurrences of the literal $x_i$ is at most four. In other words, only the four possible clause  agents (with a nonzero value for one of the four subintervals $\{e^i_2, \ldots, e^i_5\}$) can be forced to receive an empty interval when $L_{z_i}$ is set to be the entire positive side of $H_i$ (i.e., when we set $L_{z_i} = \bigcup_{k=1}^6 e^i_k$) starting with $I_{z_i} = e^i_1$. An identical argument holds when $z_i$ is assigned the negative part of $H_i$ (i.e., under the assignment $L_{z_i} = \bigcup_{k=8}^{13} e^i_k$) starting with $I_{z_i} = e^i_8$. Therefore, for each $i \in T$, updating the allocation of base agent $z_i$ (from $I_{z_i}$ to $L_{z_i}$) results in inclusion of at most four clause agents in $S$. Hence, $|S| \leq 4|T| = 4t$. 


Based on $\mathcal{L}$, we define a Boolean assignment to the variables in $\phi$, $f: \{x_1,\ldots,x_r\} \to \{0,1\}$ as follows
  \[f(x_i) \coloneqq
    \begin{cases}
      1 & if\ L_{z_i} = \bigcup_{k=8}^{13} e^i_k\\
      0 & if\ L_{z_i} = \bigcup_{k=1}^6 e^i_k
   \end{cases} \]

Observe that, by construction of $\mathcal{L}$, $f(x_i)$ is defined for all $i \in [r]$ and, for all $j \notin S \cup P$, we have $v_{a_j}(L_{a_j}) > 0$. Furthermore,  observing that, for each such $j \in [m] \setminus (S \cup P)$, the interval $L_{a_j}$ does not intersect $L_{z_i}$ we obtain the following useful fact: for each $j \in [m] \setminus (S \cup P)$, the assigned interval $L_{a_j}$ is either contained in the positive part of $H_i$ with $f(x_i) = 1$ (and literal $x_i$ occurring in clause $C_j$) or $L_{a_j}$ is contained in the negative part of $H_i$ with $f(x_i) = 0$ (and literal $\overline{x}_i$ occurring in clause $C_j$). Hence, for every index $j \in [m] \setminus (S \cup P)$, the clause $C_j$ is satisfied by the assignment $f$.  

Given that we are in the NO case, the number of satisfied clauses is at most $(1-\alpha)m$. Therefore, the cardinality of the set $[m] \setminus (S \cup P)$ is at most $(1- \alpha)m$. That is, $|S| + |P| \geq \alpha m$. Using the above-mentioned upper bound on the size of the set $S$ (i.e., $|S| \leq 4t$), we obtain the desired inequality $4t + p \geq \alpha m$. 
\end{proof}

\subsection{Proof of Claim {\ref{claim:dummyNode}}}
\label{section:proof-of-claim-3}

Recall that this claim is considered when the given Boolean formula, $\phi$, satisfies the NO case. As before, $\mathcal{C}(\phi)$ is the constructed cake-division instance and $\mathcal{I}$ is a Nash optimal allocation (in $\mathcal{C}(\phi)$) which satisfies the properties in Claim~\ref{claim:structure}.  

Furthermore, $P$ denotes the (index) set of the clause agents that share the auxiliary interval $G$ with the auxiliary agent $d$; $P \coloneqq \{j \in [m] \mid I_{a_j} \subset G\}$ and $p \coloneqq |P|$. 

\ClaimDummyNode*

\begin{proof} For an upper bound, it suffices to consider the case in which, for all $j \in P$, the density of the clause agent $a_j$ within the auxiliary interval $G$ is as high as possible, i.e., for all $j \in P$, we have $v_{a_j}(G) = 2/3$ and, by construction, this value is constant within $G$.   

By symmetry, we get that the product $v_d(I_d) \ \prod_{j \in P} v_{a_j}(I_{a_j})$ is maximized when the $p$ clause agents (indexed in $P$) are allocated intervals of the same length, say $x$. Recall that the total length of the auxiliary interval $G$ is equal to one. Hence, under the constraint that the clause agents indexed in $P$ are assigned intervals of length $x$, we get that the length of the interval allocated to the auxiliary agent $d$ is equal to $1 - px$. Since the valuation (density) of agent $d$ is constant and equal to one within $G$, the value obtained by $d$ is equal to $1-px$. Therefore, 
\begin{equation}\label{equation:dummyNode}
    v_d(I_d) \ \prod_{j \in P} v_{a_j}(I_{a_j})  \leq \max_{x \in [0,1]} \ \left(\frac{2}{3}x\right)^p (1-px)
\end{equation}
Here, the optimal value is achieved at $x^* = \frac{1}{p+1}$. Therefore, the stated claim holds. 

\end{proof}

\subsection{Proof of the bound $NSW_N \leq c(\alpha) \tau$} 
\label{section:worst-subsection-title-ever}

Recall that $\tau \coloneqq \left(1^{2r} \frac{1}{2^r} \frac{1}{3^m} 1 \right)^{\frac{1}{3r+m+1}}$ and $c(\alpha) \coloneqq 2^{-\alpha/44}$. Also, from equation \eqref{equation:NOcase}, we have
\begin{equation} \label{}
    	   NSW_N = \NSW(\mathcal{I}) \leq \left(1^{2r} \frac{1}{4^t}  \frac{1}{2^{r-t}} \frac{1}{3^{m-p}} \left( \frac{2}{3} \frac{1} {p+1}\right)^p \frac{1}{p+1} \right)^{\frac{1}{3r+m+1}}
    	\end{equation}
Dividing by $\tau$ and considering the $(3r+m+1)$th power, we get
    	\begin{align}
    	\Bigg(\frac{\NSW_N}{\tau}\Bigg)^{3r+m+1} & \leq \ \frac{3^m \ 2^r}{4^t \ 2^{r-t} \ 3^{m-p} \ (p+1)^{p+1}} \ \Bigg(\frac{2}{3}\Bigg)^p \nonumber \\
    	& = \ \frac{2^{(p-t)}}{(p+1)^{p+1}} \nonumber \\
    	& \leq \ \frac{2^{(p+ \frac{p- \alpha m}{4})}}{(p+1)^{p+1}}   \tag{$4t+p \geq \alpha m$ from Claim \ref{claim:badAgents}} \\
    	& = \ \Bigg(\frac{2^{\frac{5p}{4}}}{(p+1)^{p+1}}\Bigg) \ 2^{\frac{- \alpha m}{4}} \nonumber \\
    	& \leq \ 2^{\frac{- \alpha m}{4}} \label{ineq:nsw-n}
    	\end{align}
Here, the last inequality follows from the fact that $g(x) \coloneqq \frac{2^{\frac{5x}{4}}}{(x+1)^{x+1}}$ is a non-increasing function of $x \geq 0$ and $g(0) =1$. Additionally, the inequality $2^{\frac{- \alpha m}{4}} \leq 1$ gives us $\frac{\NSW_N}{\tau} \leq 1$.

Recall that in any instance of the $\textrm{Gap 3-SAT-5}$, the number of variables is at most three times the number of clauses, $r \leq 3m$. Therefore,

\begin{align*}
    \left(\frac{\NSW_N}{\tau} \right)^{11m} & \leq \left(\frac{\NSW_N}{\tau} \right)^{3r+m+1} \tag{since $\frac{\NSW_N}{\tau} \leq 1$ and $r \leq 3m$} \\ 
    & \leq 2^{\frac{- \alpha m}{4}} \tag{via \eqref{ineq:nsw-n}}
\end{align*}
Exponentiating both sides of the previous inequality to the power $\frac{1}{11m}$, we obtain the desired inequality $NSW_N \leq c(\alpha) \ \tau$.

\end{document}